\newif\ifdouble
\doubletrue

\newif\ifFullVersion
\FullVersiontrue  

\ifdouble
\documentclass[journal,a4paper]{IEEEtran}
\else
\documentclass[12pt, draftclsnofoot, onecolumn]{IEEEtran}
\fi
\usepackage{amsthm}
\usepackage{latexsym}
\usepackage{amssymb}
\usepackage{graphicx}
\usepackage{subcaption}
\usepackage{amsmath}
\usepackage{color}
\usepackage{cite}
\usepackage[bookmarks,colorlinks]{hyperref}
\usepackage{cleveref}
\usepackage{acronym}
\usepackage[english]{babel}
\usepackage{enumitem}
\usepackage{algpseudocode}
\usepackage[ruled,linesnumbered]{algorithm2e}
\usepackage{comment}

\SetKwInput{KwData}{\textbf{Init}}
\SetKwInput{KwDataB}{}

\theoremstyle{plain}
\newtheorem{theorem}{Theorem}
\newtheorem{lemma}{Lemma}

\newtheorem{prop}{Proposition}
\newtheorem{corollary}{Corollary}
\theoremstyle{definition}
\newtheorem{definition}{Definition}
\theoremstyle{remark}

\newcommand{\off}[1]{}

\newcommand{\cPNZ}{{\cal PNZ}}
\newcommand{\cNZ}{{\cal NZ}}

\newcommand{\rev}[1]{\textcolor[rgb]{0.00,0.00,0.00}{#1}}

\begin{document}
	\title{Serial Quantization for Sparse Time Sequences}
	\markboth{}{}
	\author{%
		\IEEEauthorblockN{Alejandro Cohen\IEEEauthorrefmark{1},
			Nir Shlezinger\IEEEauthorrefmark{2},
			Salman Salamatian\IEEEauthorrefmark{1},
			Yonina C. Eldar\IEEEauthorrefmark{2},
			and Muriel M\'edard\IEEEauthorrefmark{1}}\\
		\IEEEauthorblockA{\IEEEauthorrefmark{1}%
			Research Laboratory of Electronics, MIT, Cambridge, MA, USA,
			\{cohenale, salmansa, medard\}@mit.edu}\\
		\IEEEauthorblockA{\IEEEauthorrefmark{2}%
			Math and CS, Weizmann Institute of Science, Rehovot, Israel,
			\{nir.shlezinger, yonina.eldar\} @weizmann.ac.il}
		\thanks{
			Parts of this work were presented in the Allerton Annual Conference on Communications, Control, and Computing, July 2019.
			This project received funding from the Benoziyo Endowment Fund for the Advancement of Science, the	Estate of Olga Klein - Astrachan, the European Union’s Horizon 2020 research and innovation program under grant No. 646804-ERC-COG-BNYQ, and from the Israel Science Foundation under grant No. 0100101.
		}
		\off{\vspace{-10mm}}}
	\maketitle
	\pagestyle{plain}
	\thispagestyle{plain}
	\definecolor{NewColor}{rgb}{0,0,0.4}

\newcommand{\myVec}[1]{{\boldsymbol{#1}}}
\newcommand{\myMat}[1]{{\boldsymbol{#1}}}
\newcommand{\mySet}[1]{\mathcal{#1}}

\newcommand{\myDetVec}[1]{\myVec{\lowercase{#1}}}
\newcommand{\myRandVec}[1]{\myVec{\lowercase{#1}}}
\newcommand{\myDetMat}[1]{\myMat{\uppercase{#1}}}
\newcommand{\myRandMat}[1]{\myMat{\uppercase{#1}}}

\newcommand{\E}{\mathbb{E}}		 				
\newcommand{\myW}{{\myRandVec{W}}}			 		
\newcommand{\myY}{{\myRandVec{Y}}}			 		
\newcommand{\myZ}{{\myRandVec{Z}}}			 		
\newcommand{\myX}{{\myRandVec{X}}}			 		
\newcommand{\myV}{{\myRandVec{V}}}			 		
\newcommand{\myS}{{\myDetVec{s}}}			 		
\newcommand{\myI}{{\myDetMat{i}}}			 		
\newcommand{\myA}{{\myDetMat{a}}}
\newcommand{\myB}{{\myDetMat{b}}}
\newcommand{\myAT}{\tilde{\myA}}
\newcommand{\myBT}{\tilde{\myB}}			 					 		
\newcommand{\myAB}{\bar{\myA}}
\newcommand{\myYmat}{{\myRandMat{Y}}}			 	
\newcommand{\myYvec}{\underline{\myY}}			 	
\newcommand{\myQvec}{\underline{\myVec{q}}}			 	
\newcommand{\mySmat}{{\myMat{\Theta}}}			 	
\newcommand{\myWmat}{{\myRandMat{W}}}			 	
\newcommand{\myWvec}{\underline{\myW}}			 	
\newcommand{\Gmat}{{\myRandMat{G}}}			 		
\newcommand{\Gvec}{\underline{\myVec{g}}}			 		
\newcommand{\GmatRel}{{\myMat{g}}}			 		
\newcommand{\Hmat}{{\myRandMat{H}}}			 		
\newcommand{\Hvec}{\underline{\myVec{h}}}
\newcommand{\Dmat}{{\myDetMat{d}}}			 		
\newcommand{\Bmat}{{\myDetMat{f}}}			 		
\newcommand{\Phimat}{{\myMat{\Phi}}}			 		
\newcommand{\DmatRel}{\bar{\myDetMat{d}}}			 		
\newcommand{\BmatRel}{\bar{\myDetMat{f}}}			 		
\newcommand{\AggMat}{{\myRandMat{A}}}			 	
\newcommand{\Ymat}{\tilde{\myRandMat{Y}}}			
\newcommand{\Wmat}{\tilde{\myRandMat{W}}}			
\newcommand{\Smat}{{\myMat{\Theta}}}				
\newcommand{\myTheta}{\theta}
\newcommand{\SigW}{\sigma_W^2}						
\newcommand{\AntRatio}{\kappa}						
\newcommand{\Ncells}{n_c}							
\newcommand{\Nantennas}{\lenX}						
\newcommand{\Nusers}{\lenS}							
\newcommand{\NcellsSet}{\mySet{N}_c}				
\newcommand{\NusersSet}{\mySet{K}}				
\newcommand{\dcoeff}{d}								
\newcommand{\dcoeffRel}{d}								
\newcommand{\bcoeff}{f}								
\newcommand{\phicoeff}{\phi}								
\newcommand{\Tpilots}{\lenXtag}						
\newcommand{\TpilotsSet}{\mySet{L}}						
\newcommand{\Tdata}{\tau_d}							
\newcommand{\Tcoh}{\tau_c}							
\newcommand{\EstGmat}{\hat{\Gmat}}					
\newcommand{\ErrGmat}{\tilde{\Gmat}}				
\newcommand{\MPFunc}{\nu}								
\newcommand{\SemiCirc}{F}							
\newcommand{\UHmat}{ \myMat{M}}
\newcommand{\Sol}{s}
\newcommand{\Dist}{ \stackrel{d}{=}}
\newcommand{\AsConv}{\mathop{\longrightarrow}\limits^{\rm a.s.}}
\newcommand{\CDF}[1]{F_{#1}}
\newcommand{\Pdf}[1]{f_{ #1}}
\newcommand{\Psd}[1]{s_{#1}}
\newcommand{\Acorr}[1]{c_{#1}}
\newcommand{\PSD}[1]{\myMat{S}_{#1}}
\newcommand{\ACORR}[1]{\myMat{C}_{#1}}
\newcommand{\CorrMat}[1][ ]{\myMat{C}_{#1}}
\newcommand{\CovMat}[1]{\myMat{\Sigma}_{#1}}			
\newcommand{\CovMatExt}[1]{{\underline{\myMat{\Sigma}}}_{#1}}			
\newcommand{\maxDiag}{\sigma^2_{l}}
\newcommand{\bits}{b}
\newcommand{\SpaSize}{k}
\newcommand{\Rate}{R}
\newcommand{\Ratio}{r}
\newcommand{\AsymDist}{\mu} 
\newcommand{\lenX}{n}			 			
\newcommand{\lenZ}{P}			 			
\newcommand{\lenZT}{\tilde{\lenZ}}			 			
\newcommand{\lenZn}{m_p}
\newcommand{\lenZq}{m_q}
\newcommand{\lenSset}{\mySet{K}}			 			
\newcommand{\lenXset}{\mySet{N}}			 			
\newcommand{\lenT}{T}
\newcommand{\Quan}[2]{Q_{{#1}}^{{#2}}}
\newcommand{\LmmseMat}{\myMat{\Gamma}}
\newcommand{\LmmseMatT}{\tilde{\LmmseMat}}
\newcommand{\EmpSet}{\varnothing}
\newcommand{\DynRange}{\gamma}
\newcommand{\DynInt}[1][ ]{\Delta_{#1}}
\newcommand{\TilM}[1][ ]{\tilde{M}_{#1}}
\newcommand{\MyKappa}[1][]{\kappa_{#1}}
\newcommand{\Qnoise}{\myVec{e}}
\newcommand{\Wlevel}{\zeta}
\newcommand{\myEta}{\eta}
\newcommand{\DistG}{D_{\rm G}}
\newcommand{\MMSE}{^{\rm MMSE}}
\newcommand{\Opt}{^{\rm Opt}}
\newcommand{\op}{^{\rm o}}
\newcommand{\Ign}{^{\rm Ign}}
\newcommand{\ADC}{^{\rm HL}}
\newcommand{\sADC}{^{\rm sHL}}
\newcommand{\myObs}{\underline{\myObstag}}
\newcommand{\mySOI}{\underline{\mySOItag}}
\newcommand{\mySOIEst}{\underline{\mySOIEsttag}}
\newcommand{\myQ}{\myVec{q}}
\newcommand{\lenXtag}{L}
\newcommand{\lenS}{m}
\newcommand{\myObstag}{\myVec{y}}
\newcommand{\mySOItag}{\myVec{g}}
\newcommand{\mySOIEsttag}{\tilde{\mySOItag}}
\newcommand{\LmmseMattag}{{\LmmseMat}}
\newcommand{\eig}[1]{\lambda_{#1}}			
\newcommand{\eigT}[1]{\eig{#1}}
\newcommand{\CovYtag}{\CovMat{\myY_l}}
\newcommand{\myAtag}{\myA\op}
\newcommand{\myBtag}{\myB\op}
\newcommand{\Glevel}{\varphi}
\newcommand{\GlevelT}{\tilde{\varphi}}
\newcommand{\Plevel}{\Glevel}

\newcommand{\lenL}{l}			 			
\newcommand{\myBin}{\mySet{B}}
\newcommand{\myCodeword}{\myVec{c}}
\newcommand{\myCodewordT}{\tilde{\myVec{c}}}
\newcommand{\ScaQuant}{q}

\acrodef{bs}[BS]{base station}
\acrodef{mimo}[MIMO]{multiple-input multiple-output}
\acrodef{mac}[MAC]{multiple access channel}
\acrodef{dsp}[DSP]{digital signal processor}
\acrodef{ut}[UT]{user terminal}
\acrodef{cdf}[CDF]{cumulative distribution function}
\acrodef{pdf}[PDF]{probability density function}
\acrodef{ps}[PS]{pilot sequence}
\acrodef{se}[SE]{spectral efficiency}
\acrodef{mse}[MSE]{mean-squared error}
\acrodef{adc}[ADC]{analog-to-digital convertor}
\acrodef{dtft}[DTFT]{discrete-time Fourier transform}
\acrodef{dft}[DFT]{discrete Fourier transform}
\acrodef{nb}[NB]{narrowband}
\acrodef{dt}[DT]{discrete-time}
\acrodef{ct}[CT]{continuous-time}
\acrodef{evd}[EVD]{eigenvalue decomposition}
\acrodef{svd}[SVD]{singular valued decomposition}
\acrodef{soi}[SOI]{signal of interest}
\acrodef{awgn}[AWGN]{additive white Gaussian noise}
\acrodef{wss}[WSS]{wide-sense stationary}
\acrodef{mmse}[MMSE]{minimum \ac{mse}}
\acrodef{mi}[MI]{mutual information}
\acrodef{lmmse}[LMMSE]{linear MMSE}
\acrodef{map}[MAP]{maximum a-posteriori probability}
\acrodef{ml}[ML]{maximum likelihood}
\acrodef{isi}[ISI]{intersymbol interference}
\acrodef{snr}[SNR]{signal-to-noise ratio}
\acrodef{pc}[PC]{proper-complex}
\acrodef{cs}[CS]{compressed sensing}
\acrodef{psd}[PSD]{power spectral density}
\acrodef{ptp}[PtP]{point-to-point}
\acrodef{sinr}[SINR]{signal-to-interference-and-noise ratio}
\acrodef{pdf}[PDF]{probability density function}
\acrodef{rv}[RV]{random variable}
\acrodef{csi}[CSI]{channel state information}
\acrodef{sqrss}[SQuaTS]{serial quantization of sparse time sequences}
\acrodef{qiht}[QIHT]{quantized iterative hard thresholding}
\acrodef{fista}[FISTA]{fast iterative soft thresholding algorithm}  %

	\begin{abstract}
		Sparse signals are encountered in a broad range of applications. In order to process these signals using digital hardware, they must be first sampled and quantized using an \ac{adc}, which typically operates in a serial scalar manner.  In this work we propose a method for \ac{sqrss} inspired by group testing theory, which is designed to reliably and accurately quantize sparse signals acquired in a sequential manner using serial scalar \acp{adc}. Unlike previously proposed approaches which combine quantization and \ac{cs}, our  \ac{sqrss} scheme updates its representation on each incoming analog sample and does not require the complete signal to be observed and stored in analog prior to quantization. We characterize the asymptotic tradeoff between accuracy and quantization rate of \ac{sqrss} as well as its computational burden. We also propose a variation of \ac{sqrss}, which trades rate for computational efficiency.
		Next, we show how \ac{sqrss} can be naturally extended to distributed quantization scenarios, where a set of jointly sparse time sequences are acquired individually and processed jointly.
		Our numerical results demonstrate that \ac{sqrss} is capable of achieving substantially improved representation accuracy over previous \ac{cs}-based schemes without requiring the complete  set of analog signal samples to be observed prior to its quantization, making it an attractive approach for acquiring sparse time sequences.
	\end{abstract}

	\acresetall
	\section{Introduction}
	Quantization allows continuous-amplitude physical signals to be represented using discrete values and processed in digital hardware. Such continuous-to-discrete conversions  play an important role in digital signal processing systems \cite{gray1998quantization}. In theory, jointly mapping a set of samples via vector quantization yields the most accurate digital representation \cite[Ch. 10]{C10}. However, as such joint mappings are difficult to implement, quantization is most commonly carried out using \acp{adc}, which operate in a serial and scalar manner, namely, the analog signal is sampled and each incoming sample is sequentially mapped into a discrete representation using the same mapping \cite{kosonocky1999analog}. Since \acp{adc} operating at high frequencies are costly in terms of memory and power usage, it is often desirable to utilize low quantization rates, i.e., assign a limited number of bits per each input sample, inducing additional quantization error which degrades the  digital representation accuracy \cite[Ch. 23]{polyanskiy2014lecture}.
	
	The quantization error encountered under bit budget constraints can be mitigated by accounting for underlying structure or the system task. Such quantization systems are the focus of several recent works. For example, scalar quantization mappings designed to maximize the mutual information and Fisher information with respect to a statistically related quantity were studied in \cite{bhatt2018information} and \cite{barnes2019learning}, respectively. The work \cite{shlezinger2018hardware} showed that a quantization system using uniform \acp{adc} can approach the performance achievable using vector quantizers when the system task is not to recover the analog signal, but to estimate some lower-dimensional information embedded into it.
	This approach was extended to massive \ac{mimo} channel estimation with quantized outputs in \cite{shlezinger2018asymptotic} as well as to the recovery of quadratic functions in \cite{salamatian2019task}, while a data-driven implementation was proposed in \cite{shlezinger2019deep}. The systems proposed in \cite{shlezinger2018hardware,shlezinger2018asymptotic,salamatian2019task,shlezinger2019deep} used hybrid architectures, namely, allowed some constrained processing to be carried out in analog prior to quantization, in order to mitigate the error induced by bit-limited serial scalar \acp{adc}.
	
	Conventional quantization theory considers the acquisition of a discrete-time analog source into a digital form \cite{gray1998quantization}. In some practical applications, such as sensor networks, multiple signals are acquired in distinct physical locations, while their digital representation is utilized in some central processing device, resulting in a distributed quantization setup. The recovery of a single parameter from the acquired signals was considered in \cite{gubner1993distributed, lam1993design} and its extension to the recovery of a common source, known as the CEO problem, was studied in \cite{berger1996ceo, oohama1998rate}, see also \cite[Ch. 12]{el2011network}.  Joint recovery of sources acquired in a distributed manner was studied in \cite{shlezinger2019joint}, which focused on sampling, while \cite{saxena2006efficient, wernersson2009distributed} proposed non-uniform quantization mappings for the representation of multiple sources. Multivariate (vector) quantizers for arbitrary networks were considered in \cite{fleming2004network}.
	
	A common structure exhibited by physical signals is sparsity. Sparse signals are frequently encountered in various applications, ranging from biomedical and optical imaging \cite{wagner2012compressed, shechtman2014gespar} to radar \cite{rossi2014spatial} and communications \cite{berger2010application,  feizi2011power}. An important property of sparse signals is the fact that they can be perfectly reconstructed from a lower-dimensional projection without knowledge of the sparsity pattern. This property is  studied within the framework of \ac{cs} \cite{eldar2012compressed,duarte2011structured}, which considers the recovery of sparse signals from their lower dimensional projections.
	
	Recovery of sparse signals from quantized measurements is the focus of a large body of work \cite{jacques2013robust, boufounos20081,jacques2011dequantizing, gunturk2010sigma, kipnis2018single,boufounos2015quantization,saab2018quantization}. The most common approach studied in the literature is to first project the signal in the analog domain and then  quantize the compressed measurements, via one bit representation \cite{jacques2013robust, boufounos20081}, uniform quantization \cite{jacques2011dequantizing}, sigma-delta quantization \cite{gunturk2010sigma,saab2018quantization}, or vector source coding \cite{kipnis2018single}.
	A detailed survey and analysis of methods combining quantization and \ac{cs} can be found in \cite{boufounos2015quantization}.
	In the context of distributed systems, \ac{cs} for multiple signals acquired separately  was studied in \cite{sarvotham2005distributed,baron2009distributed,do2009distributed,patterson2014distributed,feizi2010compressive},  while \cite{shirazinia2014distributed,leinonen2018distributed} proposed vector quantization schemes for bit-constrained distributed \ac{cs}.
	%
	 Despite the similarity, there is a fundamental difference between distributed quantization of sparse signals and distributed \ac{cs} with quantized observations: In the quantization framework, the measurements are the sparse signals, while in \ac{cs} the observations are  a linear projection of the signals. Consequently, to utilize \ac{cs} methods, one must first have access to the complete signal in order to project it and then quantize, imposing a major drawback when acquiring time sequences. On the other hand, a sequential approach allows to quantize without requiring that the entirety of the signal be accessible, which is particularly relevant not only in the distributed scenarios, but also for signals sparse in time, since storing the entire time-signal in analog form is expensive. Furthermore, while \ac{cs} algorithms have been proven to achieve asymptotic recovery guarantees, their performance may be degraded in finite signal sizes. These drawbacks give rise to the need for a reliable and sequential method for quantizing and recovering sparse signals, which is the focus of this work.
	
	Here, we propose \ac{sqrss}, which is a method for quantizing and recovering discrete-time sparse time sequences utilizing standard serial scalar \ac{adc} quantizers \footnote{While \ac{adc} traditionally refers to hardware which samples and quantizes an analog signal, here we denote by \ac{adc} only the quantization end of this hardware, i.e., we assume throughout the paper a discrete time input signal.}. Our scheme is inspired by recent developments in group testing theory, and leverages coding principles initially designed for secure group testing \cite{9218939} to facilitate recovery of the time sequence. In particular, \ac{sqrss} first quantizes each sample using a scalar \ac{adc}, and uses the \ac{adc} output to update a single binary value, which in turn is used as a codeword from which the sequence can be accurately recovered with high probability. The resulting coding scheme, which quantizes the sparse signal directly and operates over the binary field, allows improved reconstruction compared to \ac{cs}-based methods, which process a quantized linear projection of the real-valued observations, while also avoiding the need to store samples in analog by sequentially updating a single register.

	We characterize the achievable accuracy of \ac{sqrss} in the asymptotically large signal size regime, showing that any fixed desirable distortion level can be achieved with an overall number of bits which  grows logarithmically in the signal dimensionality and linearly with the support size.  We then characterize the computational complexity of \ac{sqrss}, and propose a reduced complexity scheme for implementing \ac{sqrss} at the cost of degraded representation accuracy.
	
	Next, we discuss how \ac{sqrss} can be naturally applied for distributed quantization of a set of temporally jointly sparse time sequences. We begin with the case where each acquired signal is conveyed to the central unit via a direct link, representing, e.g., single-hop networks.
	Then, we show how the technique can be extended to multi-hop networks, in which the quantized data  must travel over multiple intermediate links to reach the central server, and formulate simplified network policies, dictating the behavior of each intermediate node.
	We characterize the achievable distortion of \ac{sqrss} when applied in a distributed setup,
	and derive conditions on the system parameters under which it can achieve the same distortion as when applied in a non-distributed case, assuming that there exists at least a single path to the central unit.
	
	Our numerical results demonstrate that  both  \ac{sqrss}   and its reduced complexity variation   substantially outperform the conventional approach combining \ac{cs} and quantization when applied to the digital representation of a single sparse time-sequence, as well as in distributed acquisition scenarios with jointly sparse signals. This demonstrates the potential of  \ac{sqrss}   for feasible and reliable quantization of sequentially acquired sparse time sequences.
	
	The rest of this paper is organized as follows: In Section~\ref{Sec:Preliminaries} we review some preliminaries in quantization theory and present the system model. Section~\ref{sec:sqrss} proposes \ac{sqrss} along with a discussion and an asymptotic performance analysis.  Section~\ref{sec:efficient_algorithms} presents a reduced complexity variation of \ac{sqrss}.
	In Section~\ref{sec:DistQuant} we  apply \ac{sqrss} for distributed quantization.
	 Section~\ref{sec:sims} details the simulation study, and Section~\ref{sec:conclusions} provides concluding remarks. Proofs of the  results stated in the paper are detailed in the appendix.
	
	Throughout the paper, we use boldface lower-case letters for vectors, e.g., ${\myVec{x}}$.
	Matrices are denoted with boldface upper-case letters,  e.g.,
	$\myMat{M}$.
	Sets are expressed with calligraphic letters, e.g., $\mySet{X}$, and $\mySet{X}^n$ is the $n$th order Cartesian power of $\mySet{X}$.
	The stochastic expectation is denoted by   $\E\{ \cdot \}$,  $ \bigvee$ is the Boolean OR operation, and $\mySet{R}$  is the set of real  numbers.
	All logarithms are taken to base-2.
	
	\section{Preliminaries and System Model}
	\label{Sec:Preliminaries}
	
	\subsection{Preliminaries in Quantization Theory}
	\label{subsec:Pre_Quantization}
	To formulate the quantization of sparse signals setup, we first briefly review standard quantization notions.
	We begin with the definition of a quantizer:
	\begin{definition}[Quantizer]
		\label{def:Quantizer}
		A quantizer $\Quan{M}{\lenX,\lenS}\left(\cdot \right)$ with $\log M$ bits, input size $\lenX$, input alphabet $\mySet{S}$, output size $\lenS$, and output alphabet $\hat{\mySet{S}}$, consists of:
		{\em 1)} An  encoding function $g_\lenX^{\rm e}: \mySet{S}^\lenX \mapsto \{0,1,\ldots,M-1\} \triangleq \mySet{M}$ which maps the input from $\mySet{S}^\lenX$ into a discrete index $j \in \mySet{M}$.
		{\em 2)} A decoding function  $g_\lenS^{\rm d}: \mySet{M} \mapsto \hat{\mySet{S}}^\lenS$ which maps each index $j \in \mySet{M}$ into a codeword $\myVec{q}_j \in  \hat{\mySet{S}}^\lenS$.
	\end{definition}
	The quantizer output for an input $\myS  \in \mySet{S}^\lenX$ is $\hat{\myS}  = g_\lenS^{\rm d}\left( g_\lenX^{\rm e}\left( \myS\right) \right) \triangleq \Quan{M}{\lenX,\lenS}\left( \myS\right)$. An illustration is depicted in Fig. \ref{fig:Q_System}.
	\begin{figure}
		\centering
		{\includegraphics[width = \columnwidth]{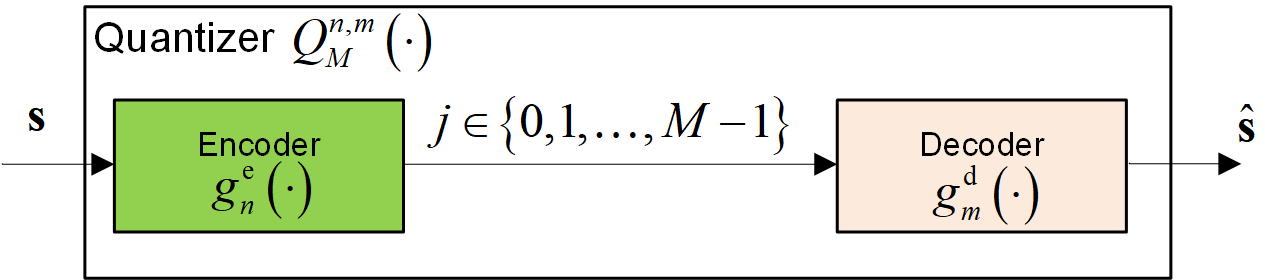}}
		\caption{Quantization system.}
		\label{fig:Q_System}
	\end{figure}
	{\em Scalar quantizers} have scalar input and output, i.e., $\lenX=\lenS=1$ and $\mySet{S}$ is a scalar space, while {\em vector quantizers} operate a multivariate input.
	When the input size and  output size are equal, namely, $\lenX=\lenS$, we write $\Quan{M}{\lenX}\left(\cdot \right) \triangleq \Quan{M}{\lenX,\lenX}\left(\cdot \right)$, while for scalar quantizers we use the notation $\Quan{M}{}\left(\cdot \right) \triangleq \Quan{M}{1}\left(\cdot \right)$.

	In the conventional quantization problem,  a $\Quan{M}{\lenX}\left(\cdot \right)$ quantizer is designed to minimize some distortion measure  $d_\lenX:\mySet{S}^\lenX\times\hat{\mySet{S}}^\lenX \mapsto \mySet{R}^+$  between its input and its output.
	The performance of a quantizer is therefore characterized using two measures: The quantization rate, defined as $\Rate \triangleq \frac{1}{\lenX}\log M$, and the expected distortion $\E\{d_\lenX\left(\myS,  \hat{\myS}  \right)\}$. For a fixed input size $\lenX$ and codebook size $M$, the optimal quantizer is given by
	\begin{equation}
	\label{eqn:OptQuantizer}
	\Quan{M}{\lenX, {\rm opt}}\left(\cdot \right) = \mathop{ \min}\limits_{\Quan{M}{\lenX}\left(\cdot \right)} \E \left\{d_\lenX\left(\myS , \Quan{M}{\lenX}\left( {\myS}  \right)\right)   \right\}.
	\end{equation}
	In the following, the distortion between a source realization $\myS$ and a reconstruction sequence $\hat{\myS}$ is defined as the  \ac{mse} of their difference, which is given by
	\begin{equation}
	\label{eqn:DisQuantizer}
	d_\lenX\left(\myS , \hat{\myS}  \right) \triangleq \frac{1}{\lenX}||\myS -\hat{\myS} ||^2.
	\end{equation}
	
	Characterizing the optimal quantizer via \eqref{eqn:OptQuantizer} and the distortion via \eqref{eqn:DisQuantizer},  as well as the optimal tradeoff between distortion and quantization rate, is in general a difficult task. Consequently, optimal quantizers are typically studied assuming either high quantization rate, i.e., $\Rate \rightarrow \infty$, see, e.g., \cite{li1999asymptotic}, or asymptotically large input size, namely, $\lenX \rightarrow \infty$, typically with stationary inputs, via rate-distortion theory \cite[Ch. 10]{C10}.
	
	Comparing high  rate analysis for scalar quantizers and rate-distortion theory for vector quantizers demonstrates the sub-optimality of serial scalar quantization. For example, for quantizing a large-scale real-valued Gaussian random vector with i.i.d. entries and sufficiently large quantization rate $\Rate$, where intuitively there is little benefit in quantizing the entries jointly over quantizing each entry independently, vector quantization notably outperforms serial scalar quantization \cite[Ch. 23.2]{polyanskiy2014lecture}.
	Nonetheless, vector quantizers are significantly more  complex compared to  scalar quantizers. One of the main sources for this increased complexity stems from the fact that vector quantizers operate on a set of analog samples. As a result, a \ac{dsp} utilizing vector quantizers to acquire an analog time sequence must store  $\lenX$ samples in the analog domain before it produces a digital representation, which may be difficult to implement, especially for large $\lenX$. Scalar quantizers, commonly used in \ac{adc} devices \cite{eldar2015sampling}, do not store data in analog as each incoming sample is immediately converted into a digital representation, and are the focus here.
	
	\subsection{System Model}
	\label{subsec:Pre_Problem}
	We consider the acquisition of a sampled time sequence $s[i]$ observed over the period $i \in \{1,\ldots,\lenT\} \triangleq \mySet{\lenT}$ into a digital representation $\hat{s}[i]$ using up to $\bits$ bits, i.e., $M = 2^\bits$ codewords.
	The signal $s[i] $ is temporally sparse with support size  $\SpaSize < \lenT$, where $\SpaSize$ is a-priori known\footnote{\rev{While we carry out our derivations and analysis assuming $\SpaSize$ is known, we only require an upper bound on it. In fact, \ac{sqrss} method support operation with erroneous knowledge of $\SpaSize$, as discussed in Section~\ref{subsec:Discussion}.}}.
	We propose a quantization system which is specifically designed to exploit this  sparsity to improve the recovery accuracy. In particular, we propose an encoder-decoder pair which utilizes tools from group testing theory to exploit the underlying sparsity of the continuous amplitude signal.
	\begin{figure}
		\centering
		{\includegraphics[width = \columnwidth]{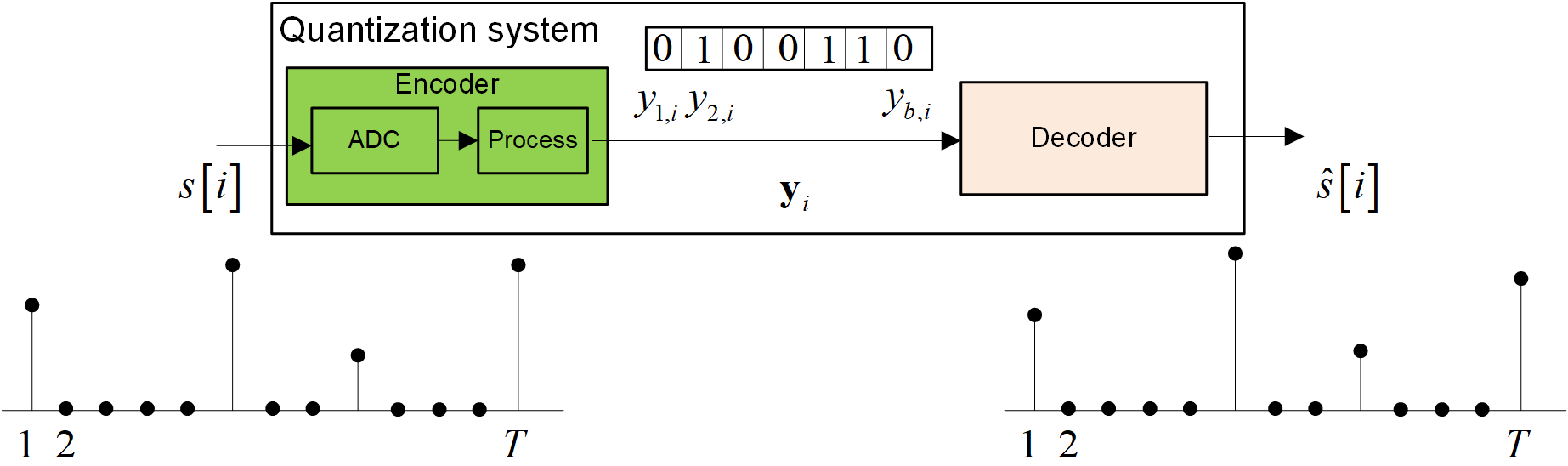}}
		\caption{Serial quantization system for sparse time sequences.}
		\label{fig:QCS_System}
	\end{figure}
	
	In order to avoid the need to store samples in analog, the system operates at each sample $s[i]$ independently. In particular, on each time instance $i \in \mySet{\lenT}$,  the encoder updates a register of $\bits$ bits, whose value upon the encoding of $s[i]$ is denoted by $\myY_i $. Once the complete time sequence is observed, i.e., $i = \lenT$, the decoder uses the digital codeword $\myY_{\lenT} $ to produce an estimate of the sequence denoted $\{\hat{s}[i]\}_{i\in\mySet{\lenT}}$. An illustration of the system is depicted in Fig. \ref{fig:QCS_System}. Since the decoder processes discrete codeword $\myY_{\lenT} $, while each $\myY_{i} $, ${i<\lenT}$ is stored only during the $i$th acquisition step, the system uses $\bits$ bits for digital representation.

	
	\section{The \ac{sqrss} System } \label{sec:sqrss}
	We next detail the proposed \ac{sqrss} system. The main rationale of \ac{sqrss} is to facilitate quantization of sparse signals using  conventional low-complexity serial scalar quantizers by utilizing group testing theory tools.
	Broadly speaking \ac{sqrss} quantizes each incoming sample using a scalar \ac{adc}. However, instead of storing this quantized value, it is used to update a $\bits$ bits codeword, which is decoded into a digital representation of the sparse signal. This approach allows to quantize each incoming sample with relatively high resolution, while using a single  $\bits$ bits register from which the digital representation of the complete signal is obtained.
	
	To properly formulate \ac{sqrss}, we first present the codebook generation in Subsection \ref{subsec:codebook}. Then, we elaborate on the \ac{sqrss} encoder and decoder structures in Subsections \ref{subsec:encoder} and \ref{subsec:decoder}, respectively. In Subsection \ref{subsec:MSE} we characterize the achievable distortion of \ac{sqrss} in the large signal size regime. Finally, in Subsection \ref{subsec:Discussion} we discuss the pros and cons of \ac{sqrss} compared to previously proposed approaches for quantizing sparse signals.
	
	\subsection{Codebook Generation} \label{subsec:codebook}
	The \ac{sqrss} system maintains a codebook used by its encoder and decoder.
	In particular, for a time sequence of $\lenT$ samples,  \ac{sqrss} uses a codebook  of $\lenL \cdot \lenT + 1$ codewords, each consisting of $\bits$ bits, where $\lenL$ is a fixed integer. We discuss the effect of $\lenL$ on the \ac{mse} and the complexity of  \ac{sqrss} in Subsection \ref{subsec:MSE}, and propose guidelines for setting its value to optimize the tradeoff between these  performance measures.
	
	Our codebook design is based on the wireless sensor coding scheme of \cite{wsn2017drivejornal}, which is inspired by recent advances in group testing theory \cite{dorfman1943detection}, and particularly the code proposed in \cite{9218939} for secure group testing.
	The objective in group testing is to identify a subset of defective items in a larger set using as few  measurements as possible. This objective can be recast as a codebook generation problem, such that for each outcome vector, i.e., a set of  measurements, it should be possible to identify the non-zero inputs \cite{dorfman1943detection}. While this setup bears much similarity to our quantization of sparse sources problem, in group testing the inputs are represented over a binary field, while in our setting the inputs can be any real value.  Consequently, the codebook here needs to be able not only to detect the indexes of the non-zero inputs, as in conventional group testing, but also to recover their value. To facilitate our design, we henceforth assume that the inputs are discretized to a set of $l+1$ different values, and show how this is incorporated into the overall scheme in the following subsections. We refer to \cref{sec:sims} for a discussion on how the parameter $l$ relates to the overall quantization rate.
	
	In particular, to formulate the codebook, we generate $\lenT \cdot \lenL \cdot \bits$ independent realizations from a Bernoulli distribution with mean value $\frac{\ln (2)}{\SpaSize} $. These realizations form $\lenL \cdot \lenT$ mutually independent codewords. The codewords are then divided into $\lenT$ bins, denoted $\myBin_i \triangleq \{\myCodeword_{j,i}\}_{j=1}^{\lenL}$, $i \in \mySet{\lenT}$, and we add to each bin the all-zero codeword denoted $\myCodeword_{0}$.
	Since $\myCodeword_{0}$ is common to all the bins, the total number of codewords is $\lenL \cdot \lenT +1$.
	The benefits of this codebook design are discussed in Subsection~\ref{subsec:Discussion}.


	\subsection{Encoder Structure} \label{subsec:encoder}
	Having generated $\lenT$ bins of $\lenL$ codewords, $\{\myBin_i\}_{i=1}^{\lenT}$, we now discuss the encoding process. To that aim, we fix some scalar quantization mapping $\Quan{\lenL+1}{}(\cdot)$ over $\mySet{R}$ with resolution $\lenL + 1$, denoted henceforth as $Q(\cdot)$ for simplicity, and let $\{\ScaQuant_j\}_{j=0}^{\lenL}$ be the set of its possible outputs.
	The specific selection of the quantization mapping represents the acquisition hardware. For example, when using the common flash \ac{adc} architecture, $Q(\cdot)$ represents a uniform quantization mapping with $\lenL + 1$ uniformly spaced decision regions \cite{kosonocky1999analog}.
	Without loss of generality, we assume that the scalar quantizer maps the input value $0$ into the discrete value $\ScaQuant_0$, namely, $Q(0) = \ScaQuant_0$.
	
	The encoding process consists of the following three stages, illustrated in Fig. \ref{fig:QCS_enc}:
	\begin{enumerate}[label={\em E\arabic*}]
		\item \label{itm:E1} Each incoming sample $s[i]$ is quantized into the discrete scalar value $Q(s[i])$. Since this same identical mapping is applied to each incoming sample in a serial manner, it can be implemented using conventional \acp{adc}.
		\item \label{itm:E2} The encoder  uses the index of the discrete value $Q(s[i])$ to select a codeword from the $i$th bin as follows: If $\ScaQuant_j = Q(s[i])$, then the selected codeword is $\myCodewordT_{i} = \myCodeword_{j,i} \in \myBin_i$.
		\item \label{itm:E3} The encoder output  $\myY_{i} $, which is initialized such that  $\myY_{0} $ is the all-zero vector, is updated by taking its Boolean OR with the selected codeword $\myCodewordT_{i} $, i.e.,
		\begin{equation}
		\label{eqn:EncOutput}
		\myY_{i}  = \myY_{i-1}  \bigvee \myCodewordT_{i}.
		\end{equation}
		Consequently, the encoder output $\myY_{\lenT} $ is given by
		\begin{equation}
		\label{eqn:EncOutput2}
		\myY_{\lenT}  = \bigvee_{i=1}^{\lenT}\myCodewordT_{i}.
		\end{equation}
	\end{enumerate}
	
	\begin{figure}
		\centering
		{\includegraphics[width = \columnwidth]{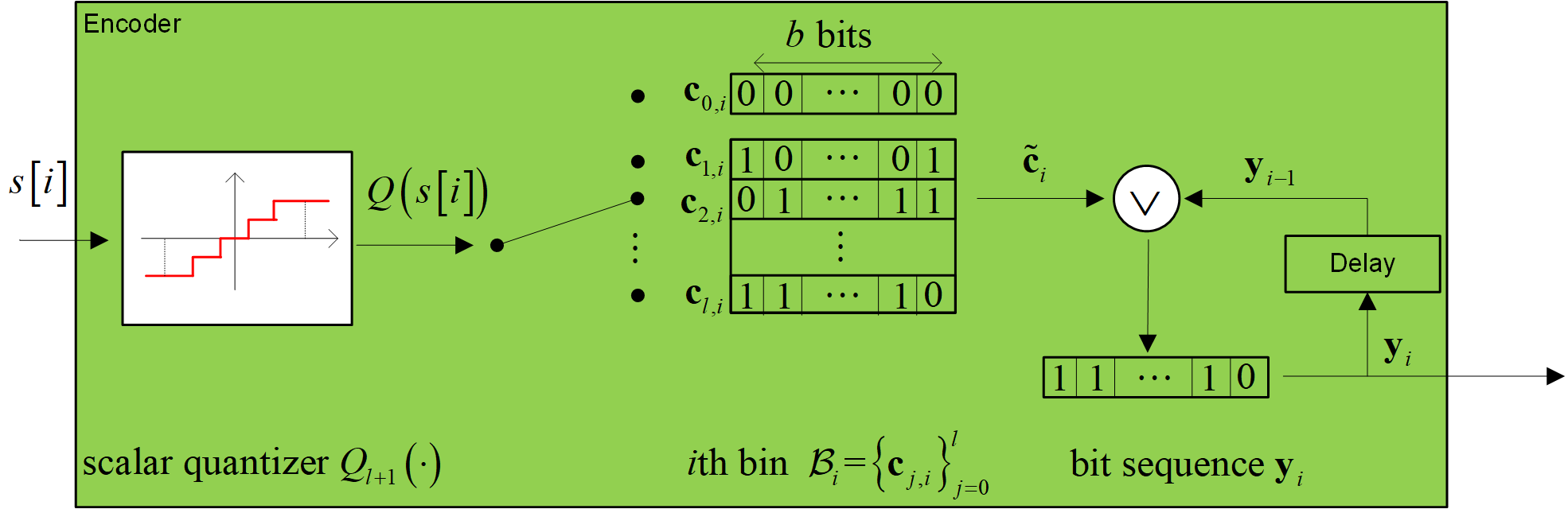}}
		\caption{Encoding process of the \ac{sqrss} system.}
		\label{fig:QCS_enc}	
	\end{figure}

	Note that only the discrete index of the quantized $Q(s[i])$, and not its actual value, affects the selection of the encoder output $\myY_{\lenT} $. Nonetheless, in Subsection \ref{subsec:MSE} we show that the selection of the output of $Q(\cdot)$, i.e., the values of $\{\ScaQuant_j\}$, and not only its partition of $\mySet{R}$ into decision regions, affect the overall \ac{mse} of  \ac{sqrss}.
	Additionally, the formulation of $\myY_{\lenT} $ via \eqref{eqn:EncOutput} implies that it can be represented using a single register of $\bits$ bits, which is updated using logical operations on each incoming sample. Consequently, while the encoder assigns a $\bits$ bits codeword to each incoming sample, the overall output size is $\bits$ and not $\lenT \cdot \bits$, thus the quantization rate is $\Rate = \frac{\bits}{\lenT}$. 
	\subsection{Decoder Structure} \label{subsec:decoder}
	The recovery of the digital representation $\{\hat{s}[i]\}_{i \in \mySet{\lenT}}$ from the output of the encoder $\myY_{\lenX}  \in \{0,1\}^{\bits}$ is based on \ac{ml} decoding. In this decoding scheme, the most likely set of $\SpaSize$ codewords are selected, from which the digital representation is obtained. To formulate the decoding process,  recall that the set $\mySet{\lenT}$ has exactly $ {\lenT \choose \SpaSize}$ possible subsets of size $\SpaSize$, representing the possible sets of non-zero entries of $\myS $. We use $\{\mySet{X}_w\}_{w \in \{1,\ldots, {\lenT \choose \SpaSize}\}}$ to denote these subsets.  The \ac{sqrss} decoder implements the following steps:
	\begin{enumerate}[label={\em D\arabic*}]
		\item \label{itm:D1} For a given encoder output $\myY_{\lenT} $, the decoder recovers a collection of $\SpaSize$ codewords $\hat{\myMat{C}}_{\mySet{X}_w} = \{\myCodeword_{j_i,i}\}_{i \in \mySet{X}_w}$, {\em each one taken from a separate bin}, for which $\myY_{\lenT} $ is most likely, namely,
		\begin{equation}
		\label{eqn:MLDef}
		\Pr\left(\myY_{\lenT}  \big|  \hat{\myMat{C}}_{\mySet{X}_w} \right) \ge \Pr\left(\myY_{\lenT}  \big|  \hat{\myMat{C}}_{\mySet{X}_{\tilde{w}}} \right), \quad \forall \tilde{w} \neq w.
		\end{equation}
		The decoder looks for both the set of $\SpaSize$ bins $\mySet{X}_w$ as well as the selection of the codeword for each bin, i.e., the selection of codeword index $j_i$ within the $i$th bin, $i \in \mySet{X}_w$, which maximize the conditional probability \eqref{eqn:MLDef}.
		\item \label{itm:D2} The decoder recovers $\{\hat{s}[i]\}$ from  $\hat{\myMat{C}}_{\mySet{X}_w} = \{\myCodeword_{j_i,i}\}_{i \in \mySet{X}_w}$ by setting its $i$th entry, denoted $\hat{s}[i]$, to be $\hat{s}[i] =\ScaQuant_{j_i} $ for  each $i \in \mySet{X}_w$ and  $\hat{s}[i] = \ScaQuant_0$ for $i \notin \mySet{X}_w$.
	\end{enumerate}
	The \ac{ml} decoder scans $\binom{\lenT}{\SpaSize}(\lenL)^\SpaSize$ possible subsets of codewords in the codebook, i.e., the $\binom{\lenT}{\SpaSize}$ possible bins corresponding to indexes which may contain non-zero values, and the $\lenL$ codewords in each such bin. For every scanned subset of codewords, the decoder compares the Boolean OR of each subset which contains $\SpaSize$ codewords to the quantized register $\myY_{\lenT} $. Since the length of each codeword is $\bits$, the computational complexity   is on the order of $	\mathcal{O}\big(\binom{\lenT}{\SpaSize}\lenL^\SpaSize \SpaSize \bits\big)$ operations.
	%
	
	While the decoding process described above may be computationally complex, it essentially implements a one-to-one mapping from $\myY_{\lenT} $ to  $\{\hat{s}[i]\}$, and can thus be implemented using a standard look-up table. In \Cref{sec:efficient_algorithms}, we present a sub-optimal low-complexity \ac{sqrss} decoder. 
	
	In the following subsection we study the achievable performance, in terms of the tradeoff between quantization rate and distortion, of the proposed \ac{sqrss} system.
	\off{
		\textcolor{red}{Alejandro - here is something we should think about: it is not clear from the description of the decoder how the complexity is affected by the setting of $\lenL$. In fact, it looks as though what really matters is the size of the lookup table which is determined by $\bits$. This is the case in standard source coding where complexity is related to the quantization rate, i.e., the number of codewords. If we indeed think that $\lenL$ is related to the complexity and not only to the tradeoff between rate and distortion, we should emphasize this dependence, and preferably quantify it. I suggest we discuss this the next time we talk.}
	}
	
	\subsection{Achievable Performance}  \label{subsec:MSE}
	In order to study the achievable performance, we first note that the \ac{sqrss} encoder and decoder are designed to recover the output of the scalar quantizer $Q(\cdot) = \Quan{\lenL+1}{}(\cdot)$. Therefore, when the \ac{sqrss} decoder detects the correct set of codewords, the distortion is determined by the scalar quantizer and its resolution, which is dictated by the  parameter $\lenL$.
	
	To formulate this distortion, define the overall average \ac{mse} of the scalar quantizer via
	\begin{equation}
	\label{eqn:ScaMSE2}
	D_{\lenT}(\lenL) \triangleq \frac{1}{\lenT}\sum\limits_{i=1}^{\lenT} \mathbb{E}\left[ |s[{i}] - \Quan{\lenL+1}{}(s[i])  |^2\right].
	\end{equation}
	The average \ac{mse} \eqref{eqn:ScaMSE2} is determined by the scalar quantization mapping $Q(\cdot)$  and the distribution of the time sequence $\{s[i]\}$.
	It represents the accuracy of applying  $Q(\cdot)$ directly to the sequence $\{s[i]\}$ without using any additional processing, thus operating at quantization rate of $\log(\lenL +1)$ bits per input sample. \ac{sqrss} with rate $\Rate$, which, as we show next, can be much smaller than $\log(\lenL +1)$,  is capable of achieving this average \ac{mse} when its decoder  successfully recovers the correct set of codewords. A sufficient condition for successful recovery in the limit of asymptotically large inputs, and thus for~\eqref{eqn:ScaMSE2} to be achievable, is stated in the following theorem:
	\begin{theorem}\label{direct theorem}
		The \ac{sqrss} system applied to a sparse signal $\{s[i]\}_{i\in\mySet{\lenT}} $ with support size $\SpaSize=\mathcal{O}(1)$ achieves the average \ac{mse} $D_{\lenT}(\lenL)$ given in \eqref{eqn:ScaMSE2} in the limit $\lenT \rightarrow \infty$ when for some $\varepsilon>0$, the quantization rate $\Rate$ satisfies:
		\begin{eqnarray}\label{eq:reduce_hw}
		\Rate \ge \Rate_\varepsilon(\lenL)  \triangleq \max_{1 \leq i \leq \SpaSize}\frac{(1+\varepsilon)\SpaSize}{i\cdot \lenT}\log\left( \binom{\lenT-\SpaSize}{i}\cdot \lenL^i\right).
		\end{eqnarray}
	\end{theorem}
	
	{\em Proof:} The proof is given in the appendix.
	
	\smallskip 	
	Theorem \ref{direct theorem} implies that, as $\lenT$ increases, if the number of bits is $\bits=\Rate \cdot \lenT$ where $\Rate$ satisfies \eqref{eq:reduce_hw}, then the average error probability in detecting the \ac{sqrss} codewords approaches zero (decaying exponentially with $\lenT$) and thus the \ac{sqrss} system achieves the average \ac{mse} $D_{\lenT}(\lenL)$ given in \eqref{eqn:ScaMSE2}.
	
	Note that the average \ac{mse} $D_{\lenT}(\lenL)$ and the corresponding quantization rate $\Rate_\varepsilon(\lenL) $ both depend on the auxiliary parameter $\lenL$. The dependence of $D_{\lenT}(\lenL)$ on $\lenL$ is obtained from the quantization mapping used, as well as the distribution of the input $\{s[i]\} $.
	For example, when the samples of $\{s[i]\} $ are identically distributed with \ac{pdf} $\Pdf{s}(\cdot)$, then, using the Panter-Dite approximation \cite{panter1951quantization}, the optimal (non-uniform) scalar quantizer in the fine quantization regime achieves the following average \ac{mse}:
	\begin{equation}
	\label{eqn:OptDist}
	D_{\lenT}(\lenL) \approx \frac{1}{12}2^{-2 \log(\lenL + 1)}\Big( \int\limits_{\alpha \in \mySet{R}} \Pdf{s}^{1/3}(\alpha) d\alpha\Big)^3.
	\end{equation}	
	The average \ac{mse} in \eqref{eqn:OptDist} implies that the achievable distortion using scalar quantizers, including conventional architectures such as uniform quantization mappings, can be made arbitrarily small by increasing the resolution  $ \log(\lenL + 1)$.
	
	While the average \ac{mse} directly depends on the quantization mapping, the rate $\Rate_\varepsilon(\lenL)$ is invariant to the setting of $Q(\cdot)$, and is obtained as the maximal value of the right hand side of \eqref{eq:reduce_hw}. To avoid the need to search for the maximal value in  \eqref{eq:reduce_hw}, we state an upper bound on $\Rate_\varepsilon(\lenL)$ in the following corollary:
	\begin{corollary}
		\label{cor:BoundRate}
		The quantization rate  $\Rate_\varepsilon(\lenL)$ in \eqref{eq:reduce_hw} satisfies
		\begin{equation}
		\label{eqn:BoundRate}
		\Rate_\varepsilon(\lenL) \leq  (1+\varepsilon) \frac{\SpaSize}{\lenT} \log \left( \lenT\cdot\lenL\right) .
		\end{equation}
	\end{corollary}
	
	\begin{IEEEproof}
		The corollary is obtained by substituting in  \eqref{eq:reduce_hw} the upper bound $\log\binom{\lenT-\SpaSize}{i} \le i\log T$, which follows from Stirling's approximation \cite{C10,chan2014non}.
	\end{IEEEproof}

	%
	%
	\smallskip
	We note that when $\SpaSize=\mathcal{O}(1)$, the upper bound \eqref{eqn:BoundRate} tends to zero as $\lenT$ grows for any fixed $\lenL$. Consequently, for large $\lenT$, \ac{sqrss} requires significantly smaller quantization rates to achieve  $D_{\lenT}(\lenL) $ compared to directly applying $Q(\cdot)$ to $\{s[i]\}$, which requires a rate of $\log (\lenL + 1)$ to achieve the same average \ac{mse}. This gain, which demonstrates the ability of \ac{sqrss} to exploit the underlying sparsity of $\{s[i]\}$, is also observed in the simulations study presented in Section~\ref{sec:sims}. 
	
	Corollary \ref{cor:BoundRate} can be used to determine the quantization rate for achieving a desirable \ac{mse} for a given family of  scalar quantization mappings: The  parameter $\lenL$ is set to the minimal value for which  $D_{\lenT}(\lenL) $ is not larger than the desirable distortion. Next, using the resulting $\lenL$, the quantization rate can be obtained using the right-hand side of \eqref{eqn:BoundRate}. Theorem \ref{direct theorem} guarantees that, for large input size $\lenT$, the desirable distortion is achievable when using \ac{sqrss} with the selected quantization rate.
	In fact, in the numerical study presented in Section \ref{sec:sims} we demonstrate that,  by properly tuning $\lenL$, the proposed system can achieve substantial \ac{mse} gains over previously proposed approaches for quantizing sparse time sequences.

	The bound on the quantization rate required to approach $D_{\lenT}(\lenL) $ given in Corollary \ref{cor:BoundRate} can also be used to characterize the asymptotic growth rate of the number of quantization bits used by the \ac{sqrss} system, $\bits$, as stated in the following corollary:

	\begin{corollary} \label{cor:asympt}
		The \ac{mse} $D_{\lenT}(\lenL) $ can be approached as $\lenT$ increases when the number of quantization bits $\bits$ grows as
		\begin{equation}
		\label{eqn:asympt}
		\bits = \mathcal{O} \left(\SpaSize \log \lenT +\SpaSize\log \lenL\right).
		\end{equation}
	\end{corollary}
	
	Corollary \ref{cor:asympt} implies that,  besides the obvious linear dependence in $\SpaSize\log \lenL$, the required number of bits grows proportionally to a logarithmic factor of $\lenT$, which depends on the sparsity pattern size $\SpaSize$.
	A similar asymptotic growth in the number of bits, i.e., proportional to $\SpaSize \log \lenT$, was also shown to be sufficient to achieve a given distortion when using \ac{cs}-based methods in \cite[Thm. 2]{jacques2013robust}. However, our numerical study presented in Section~\ref{sec:sims} demonstrates that despite the similarity in the asymptotic growth, when $\bits$ is fixed, \ac{sqrss} achieves improved reconstruction accuracy compared to \ac{cs}-based techniques.
	
	Substituting \eqref{eqn:asympt} in the \ac{ml} decoding complexity in Subsection~\ref{subsec:decoder} allows us to characterize the computational burden of the \ac{ml} decoder, as stated in the following corollary:
	\begin{corollary} \label{cor:complexity}
		 \ac{sqrss}  with the \ac{ml} decoder detailed in Subsection \ref{subsec:decoder} is capable of achieving the \ac{mse} $D_{\lenT}(\lenL) $ \eqref{eqn:ScaMSE2} in the limit $\lenT \rightarrow \infty$ with a computational complexity on the order of $\mathcal{O}\left(\binom{\lenT}{\SpaSize}\lenL^\SpaSize \SpaSize^2 \log \lenT + \binom{\lenT}{\SpaSize}\lenL^\SpaSize \SpaSize^2 \log \lenL\right)$ operations.
	\end{corollary}
	The complexity of the \ac{sqrss} decoder is significantly affected by the size of the sparsity pattern $\SpaSize$ in a much more dominant manner compared to its dependence on the signal size $\lenT$, and the resolution of the scalar quantizer $\lenL$. While this implies that the \ac{sqrss} system is most computationally efficient for highly sparse inputs, the proposed mechanism is applicable for any size of the sparsity pattern. \off{\rev{Anyway, unlike traditional schemes, e.g., classical group testing and CS approaches, where the sparsity is defined as the ratio between $k$ and $T$, using the group testing for non-binary inputs as we propose in this work, the sparsity is defined as the ratio between $k$ and $T\cdot l$. Hence, for sufficiently large $l$, the samples are considered as sparse even when $k/T \rightarrow 1$. \off{For example, as presented in Fig~\ref{fig:non_zero_samples}, for $T=50$ and $k=30$. Where $k=T$, empirically CoMa decoder gets approximately similar MSE results as with scalar quantization.}}}

    \rev{We note that SQuaTS is geared towards low rate and low resolution scenarios, where one typically has much to gain by incorporating coding schemes in quantization over merely utilizing serial scalar ADCs. As shown in  \eqref{eqn:asympt}, the codeword length scales logarithmically with $l$ and $T$. If the codewords are too long due to design issues, e.g., high values of $l$ yielding fine resolution quantization are desired, fragmentation as proposed in \cite{wsn2017drivejornal} can be used. That is, in the SQuaTS system, fragmentation by dividing the set of input bits $T$ or set the levels $l$ into two or more groups.\off{Doing so divides the set of input bits $T$ or set of the levels $l$ is partitioned into two or more groups.} This operation reduces the memory size and the decoding complexity as needed. Furthermore, when the sparsity level is approximately identical among the groups, as is the case for large groups with i.i.d. inputs or in the presence of prior knowledge of structured sparsity, fragmentation does not compromise  the performance of the proposed SQuaTS system, as we numerically demonstrate in Fig.~\ref{fig:non_zero_samples} and Fig.~\ref{fig:high_k} in Section~\ref{sec:sims}.}
	
	\subsection{Discussion}  \label{subsec:Discussion}
	We next discuss the practical aspects of this method and its rationale. In particular, we first detail the benefits which stem from the \ac{sqrss} architecture and compare it to related schemes for quantizing sparse signals, such as direct application of scalar quantizers as well as compress-and-quantize \cite{jacques2011dequantizing,jacques2013robust, boufounos20081,  gunturk2010sigma,saab2018quantization}. Then, we elaborate on the relationship between \ac{sqrss} and group testing theory.
	
	\subsubsection{Practical benefits and comparison with related schemes}
	 \ac{sqrss}  is specifically designed to utilize scalar \acp{adc} in a serial manner. The resulting structure can be therefore naturally implemented using practical  \ac{adc} architectures \cite{kosonocky1999analog}. Moreover, \ac{sqrss} is tailored to exploit an underlying sparsity of the input signal. Straight-forward application of a serial scalar \ac{adc} requires $\lenT \cdot \log (\lenL + 1)$ bits to achieve the distortion $D_{\lenT}(\lenL) $ in \eqref{eqn:ScaMSE2}. Our proposed \ac{sqrss}, which exploits the sparsity of the input by further encoding the \ac{adc} output in a serial manner,  requires $\bits = \mathcal{O} \left(\SpaSize \log (\lenT\lenL)\right)$ bits to achieve the same \ac{mse}, as follows from Corollary \ref{cor:asympt}.
	This implies that for highly sparse signals, i.e., when $\SpaSize \ll \lenT$, \ac{sqrss}  significantly reduces the number of bits while utilizing  scalar \acp{adc} for acquisition, by introducing an additional encoding applied in a serial manner at its output.	
	The resulting approach thus bears some similarity to previously proposed universal quantization methods which are based on applying entropy coding to the output of a quantizer. See \cite{ziv1985universal} for scalar quantizers and \cite{zamir1992universal} for vector quantizers. Indeed,  since the codewords representing the quantized value are generated according to a Bernoulli distribution with mean value $\frac{\ln (2)}{\SpaSize}$ and the outcome $\myY _\lenT$ is the Boolean OR of $\SpaSize$ inputs, it can be shown that its entries approach being independent and equally distributed on the set $\{0,1\}$ for large values of $\SpaSize$, namely, the optimal lossless encoded representation, as achieved using entropy coding \cite[Ch. 5]{C10}. Nonetheless, to apply conventional entropy coding, one must first quantize all the entries of the input (or at least a large block of input entries) before applying the encoding process, requiring a large number of bits to store and represent this quantized block. \ac{sqrss}, which is specifically designed to operate in a serial manner, updates the same $\bits$-bits register on each incoming sample, thus avoiding the need to store the output of the serial scalar \ac{adc} $Q(\cdot)$ prior to its encoding.
	
	Arguably the most common approach considered in the literature for quantization of sparse signals is based on \ac{cs} techniques. In these methods, a  sensing matrix is used to linearly combine  the sparse signal into a lower-dimensional vector, which is then quantized, either using optimal vector quantization, as in  \cite{kipnis2018single}, or more commonly, via some scalar continuous-to-discrete mapping, as in \cite{jacques2011dequantizing,jacques2013robust, boufounos20081,  gunturk2010sigma}. When the input signal is a sequentially acquired  time sequence, as considered here, such \ac{cs} based techniques need to store the incoming samples in the analog domain prior to their combining using the sensing matrix\footnote{One may also store only the lower-dimension compressed vector and update its entries on each incoming input sample. Yet, this approach still requires the storage of a large amount of samples in analog as quantization can only be carried out once the complete signal is compressed.}. This requirement, which does not exist for our proposed \ac{sqrss}, limits the applicability of these proposed schemes, especially for long time sequences, i.e., in the regime typically considered in the literature. It should be stressed that \ac{cs}-based methods assumes a \emph{simple} acquisition, i.e. conventionally linear, at the expense of a more complex decoding process. \ac{sqrss} on the other hand, is a more involved encoding scheme which is tailored to the task of serial acquisition of sparse signals.
	
	\rev{An additional benefit of \ac{sqrss} compared to \ac{cs}-based methods, stems from its usage of binary codebooks for compression.}  \ac{sqrss} operates directly on $\{s[i]\}$, and not on its lower dimensional projections as in \ac{cs}, assign to binary codewords which originate from group testing theory, based on the value of $\{s[i]\}$, and more precisely, on $\{Q(s[i])\}$. By doing so, \ac{sqrss} achieves improved immunity to measurement errors compared to operating over fields of higher cardinality. This benefit is translated to more accurate digital representations, as numerically  demonstrated in   Section \ref{sec:sims}.
	
	\rev{Our analysis of \ac{sqrss} is carried out assuming that the value of $k$ used in the design of the quantization is the true sparsity level, as detailed in Section~\ref{subsec:Pre_Problem}. However, \ac{sqrss} is applicable and its performance guarantees hold also when only an upper bound on $k$ is known, and the actual sparsity pattern is smaller. This is a common assumption in the group testing literature \cite{macula1999probabilistic}, and it  was shown that  $k$ can be estimated in real-time with $O(\log T)$ bits \cite{damaschke2010bounds,damaschke2010competitive}.  In general, the sparsity assumption, i.e., $k \ll T$, allows the incorporation of group testing tools to yield accurate and computationally feasible serial quantization. When the number of non-zero samples is higher than the value of $k$ used to design the code, the error probability will increase. However, if the quantization rate of the code defined is higher than the sufficiency rate (given in Theorem~\ref{direct theorem}), the decoder will not fail drastically. Only some degradation in the MSE results are obtained, as we numerically demonstrate in Section~\ref{subsec:SimSingle}.}
	
	\rev{The codebook used by the proposed SQuaTS system, as defined in Section III-A, is generated randomly. When $b$ is large, the probability of repetition is small \cite{atia2012boolean}. In the case that $b$ is small, one can choose in the codebook generation stage "typical codewords", namely, only codewords without repetition to avoid errors in the recovery at the decoding process. We note that the numbers of non-zeros bits in the codewords is dependent on $b$ rather than on the sparsity level  $k$. On average, the $nTl$ required codewords in SQuaTS system have $p \cdot b \approx  \frac{\ln2}{k} \cdot (1+\varepsilon) k \log_2 nTl \triangleq |c(1)|$ non-zeros bits. Hence, the number of possible codewords is given by, $\binom{b}{|c(1)|} \geq \left(\frac{b}{|c(1)|}\right)^{|c(1)|}$, and thus to be able to generate sufficient codewords without repetitions in the codebook, it is required that
    \[
        nTl \leq \left(\frac{b}{|c(1)|}\right)^{|c(1)|},
    \]
    for any $k,n,l,T$ and some $\varepsilon > 0 $. Rearranging terms in the inequality results in
   \begin{equation}
   \label{eqn:EpsPrime}
        \frac{1}{\ln2\log_2(\frac{k}{\ln2})}-1\leq \varepsilon^{\prime},
   \end{equation}
    where $\varepsilon = \varepsilon^{\prime} + \varepsilon^{\prime\prime}$ and  $0 \leq \varepsilon^{\prime} < \varepsilon$. In Fig.~\ref{fig:varepsilon_k} it is numerically demonstrated that for any $k\neq 1$ there are sufficient possible codewords for any $n,l,T$ with $\varepsilon^{\prime} = 0$. For $k=1$ it is required to increase the size of $b$ by $\varepsilon^{\prime}$ to obtain sufficient possible codewords without repetitions.}

    \begin{figure}
	    \centering
	    {\includegraphics[trim=0cm 0.0cm 0cm 0.0cm, width = 1 \columnwidth]{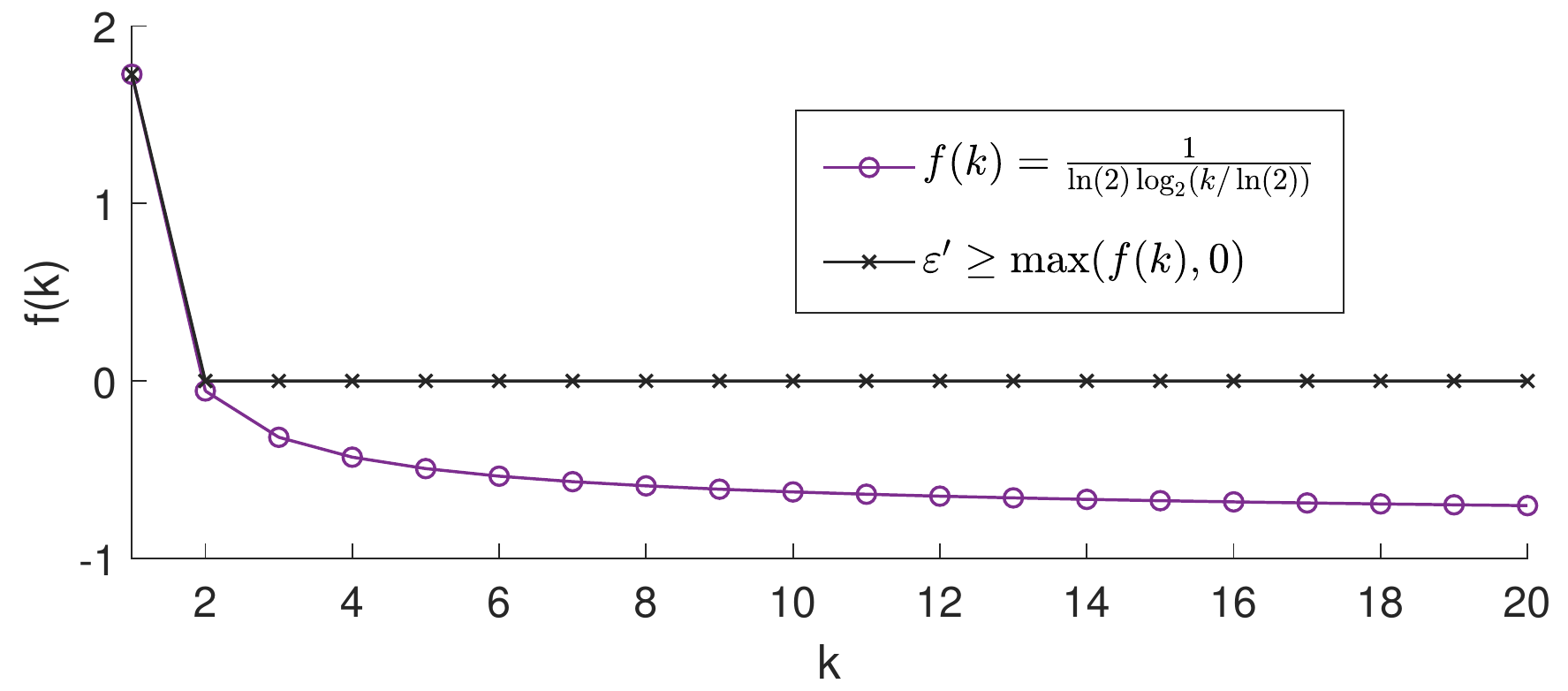}}
	    \caption{\rev{Setting of $\varepsilon^{\prime}$ in \eqref{eqn:EpsPrime} required to obtain sufficient possible codewords without repetitions.}}
	    \label{fig:varepsilon_k}
    \end{figure}
	
	A possible drawback of \ac{sqrss} compared to \ac{cs}-schemes stems from the fact that \ac{sqrss} is designed assuming that the signal $\{s[i]\} $ is sparse, i.e., that at most $\SpaSize$ of its entries are non-zero.  \ac{cs} methods are commonly capable of reliably recovering signals which are sparse in an alternative domain, namely, when there exists a non-singular matrix $\myMat{P}$ such that $\myMat{P}\myS $ is sparse, where $\myS$ is the $\lenT \times 1$ vector representation of $\{s[i]\} $. It is noted though that \ac{sqrss} can still be applied to such signals by first projecting the signal using the matrix $\myMat{P}$, resulting in a sparse signal which can be represented using \ac{sqrss}. Such application however requires the entire signal to be first acquired, as is the case with conventional \ac{cs} methods.
	
	Finally, we note that while \ac{cs} based techniques typically require the ratio between the sparsity pattern size $\SpaSize$ and the input dimensionality $\lenT$ to be upper bounded, our proposed \ac{sqrss} can be applied for any ratio between $\SpaSize$ and $\lenT$. \off{It should be emphasized though that when $\SpaSize$ is not sufficiently smaller than $\lenX$, the proposed \ac{sqrss} may no  longer be more efficient (in terms of number of quantization bits required to achieve the \ac{mse} $D_{\lenT}(\lenL) $) compared to a straight-forward application of a serial scalar quantizer $Q(\cdot)$.\footnote{\textcolor{red}{Alejandro - can you please verify this statement? Can we even quantify when does our \ac{sqrss} become less effective than just quantizing the input using a serial scalar quantizer?}}}
	
	\subsubsection{Relationship to group testing theory}
	As mentioned in Subsection \ref{subsec:codebook}, the \ac{sqrss} code construction is inspired by codebooks designed for  group testing.
	Group testing first originated from the need to identify a small subset $\SpaSize$ of infected draftees with syphilis from a large set of a population with size $\lenT$, using as few pool measurements $\bits$ as possible. Thus, group testing measurements, i.e., codewords, are designed such that given an outcome vector of size $\bits$, one should be able to identify the defective items, namely the non-zero inputs. As discussed in Subsection \ref{subsec:codebook}, a fundamental difference between our setup and conventional group testing stems from the fact that while in group testing the inputs are represented over the binary field, in our setting the inputs are the quantized values $\{Q(s[i])\}_{i=1}^{\lenT}$ whose alphabet size is $\lenL + 1$. Our code construction overcomes this difference by exploiting recent code designs targeting extended group testing models, and in particular, those considered in \cite{wsn2017drivejornal} and in the secure group testing framework \cite{9218939}.
	
	The resulting group testing based code design leads to a compact and accurate digital representation. In particular, due to the binning structure of the code suggested, when $\SpaSize$ inputs are different from zero there are only $\binom{\lenT}{\SpaSize}l^\SpaSize$ possible subsets of codewords from which the output of the encoder is selected. For comparison, in a naive codebook which assigns a different codeword to each quantized input value without binning, there are $\binom{\lenT \lenL}{\SpaSize}$ possible subsets. This  significantly reduces the number of bits required in the outcome vector.
	
	Finally, we note that the construction of the suggested code does not depended on the distribution of the input signal, which is similar to universal quantization methods \cite{ziv1985universal}. In fact, the distribution of $\{s[i]\} $ only affects the \ac{mse} induced by the serial scalar quantizer $Q(\cdot)$. The codebook presented in Subsection \ref{subsec:codebook} is  designed to allow reliable reconstruction under the worst case scenario, i.e., the setting in which $\{Q(s[i])\}_{i=1}^{\lenT}$  are i.i.d. uniformly distributed.
	Intuitively, the quantization rate required to achieve the \ac{mse} $D_{\lenT}(\lenL) $ can be further reduced by exploiting a-priori information on the input distribution. This approach was considered for the original group testing problem with, e.g., Poisson priors in  \cite{emad2014poisson}. We leave  investigation of this approach for future study.
	
	\bigskip
	
	\ifFullVersion
	\section{Efficient Decoding Algorithm}\label{sec:efficient_algorithms}
	The \ac{sqrss} system detailed in the previous section and its performance analysis rely on \ac{ml} decoding. In particular, the  decoder uses an \ac{ml} approach to identify the sub-set of $\SpaSize$ samples of the sparse signal which may be non-zero, along with their corresponding quantized values. Such a decoding algorithm suffers from high computational complexity, as noted in Corollary \ref{cor:complexity}.
	To overcome this drawback of \ac{sqrss}, in this section we propose a decoding algorithm based on the Column Matching (CoMa) method used in \cite{chan2014non} for group testing setups.  The proposed algorithm is presented in Subsection \ref{subsec:Coma}. In Subsection \ref{subsec:ComaAnalysis} we analyze the algorithm performance and discuss its benefits.
	
	\subsection{CoMa Decoder}  \label{subsec:Coma}
	As mentioned above, our proposed decoding algorithm is based on the CoMa method \cite{chan2014non}.
	Broadly speaking, unlike the \ac{ml} decoder, which looks for the set of $\SpaSize$ codewords from different code bins which are most likely to correspond to the binary vector $\myY_{\lenT}$, CoMa decoder attempts to match a codeword from each bin to $\myY_{\lenT}$ separately. Replacing the joint search  for a set of codewords with a separate examination of each codeword significantly reduces the computational burden, at the cost of degraded decoding accuracy, as we show in Subsection \ref{subsec:ComaAnalysis}. The resulting decoder operates with the same code construction and encoder mapping as described for \ac{sqrss} in   Subsections \ref{subsec:codebook} and \ref{subsec:encoder}, respectively, thus maintaining the sequential operation and natural implementation with practical \acp{adc} of \ac{sqrss}.
	
	In particular, given an encoder output $\myY_{\lenT}$, the CoMa decoder consists of two stages: First, it scans the codebook $\{\mySet{B}_i\} = \{\myCodeword_{j,i}\}$, removing all codewords which could not have resulted in $\myY_{\lenT}$. Since the encoding procedure, and specifically, step \ref{itm:E3} detailed in Subsection \ref{subsec:encoder}, is based on a logical OR operation between the selected codewords from each bin, any codeword which has a non-zero entry in an index of zero entry of  $\myY_{\lenT}$ could not have been used in the encoding of $\myY_{\lenT}$.
	Once this elimination stage is concluded, the resulting set of possible codewords, which we denote by $\mySet{C}$, is used to generate the digital representation $\{\hat{s}[i]\}$. Specifically, for every remaining codeword $\myCodeword_{j,i}$ in $\mySet{C}$, the decoder sets $\hat{s}[i]$ to be the quantized value assigned to $\myCodeword_{j,i}$, i.e., $\hat{s}[i] = q_j$.
	The time instances $i \in \mySet{\lenT}$ for which there is no codeword in  $\mySet{C}$ are assumed to have originated from the zero codeword $\myCodeword_{0}$, and are thus set to $q_0$.
	If the remaining set of codewords $\mySet{C}$ contains several codewords from the same bin, i.e., $\exists j_1 \neq j_2$ such that $\myCodeword_{j_1,i} \in \mySet{C}$ and $\myCodeword_{j_2,i} \in \mySet{C}$, then one of these codewords is randomly selected as the one used to generate the corresponding recovered sample $\hat{s}[i]$. The decoding method is summarized as Algorithm \ref{CoMAalgo}.

%
%
%
%
%
	\begin{algorithm}[t!]
		\SetKwInOut{Input}{Input}
		\caption{CoMa Decoding.\label{CoMAalgo}}
		\small
			\Input{ $ \myY_\lenX  = (y_{1,\lenX}, \ldots, y_{\bits,\lenT})$, codebook $\{\myCodeword_{j,i}\}$.}
			\KwData{ $\mySet{C} \leftarrow \{(j,i) :  j \in\{1,\ldots,\lenL\}, i \in \mySet{\lenT}\}$.}
			\For{$i_b = 1$ to $\bits$}
			{
				\If{$\myY_{i_b,\lenT} = 0$}
				{
					$\mySet{C} \leftarrow \mySet{C} \backslash \{(j,i): (\myCodeword_{j,i})_{i_b} = 1\}$\;
				}
			}
			\For{$i=1$ to $\lenT$}
			{
				\uIf{$\exists j_i$ such that $(j_i,i) \in \mySet{C}$}
				{
					$\hat{s}[i] \leftarrow \ScaQuant_{j_i}$\;
				}
				\Else
				{
					$\hat{s}[i] \leftarrow \ScaQuant_{0}$\;
				}
			}
			\KwOut{Recovered time sequence  $\{\hat{s}[i] \} $.}
	\end{algorithm}

	\subsection{Analysis and Discussion}  \label{subsec:ComaAnalysis}
	The CoMa decoder is based on examining each codeword one-by one, which is less computationally complex compared to the straight forward  \ac{ml} evaluation, at the cost of reduced performance. 	
	As we show next, Algorithm \ref{CoMAalgo} requires a larger quantization rate  $\Rate$ to guarantee that the \ac{mse} $D_{\lenT}(\lenL)$ is achievable for any fixed $\lenL$ compared to the \ac{ml} decoder detailed in Subsection \ref{subsec:decoder}, thus trading computational burden for quantization rate.
	The performance of the \ac{sqrss} system using the CoMa decoder is stated in the following proposition:
	
	\begin{prop}
		\label{pro:ComaRate}
		The \ac{mse} $D_{\lenT}(\lenL) $ is achievable by \ac{sqrss} with the CoMa decoder in the limit $\lenT \rightarrow \infty$ when for some $\varepsilon>0$, the quantization rate $\Rate$ satisfies:
		\begin{equation}
		\label{eqn:ComaRate}
		\Rate \ge \bar{\Rate}_{\varepsilon}(\lenL) \triangleq  \frac{(1+\varepsilon) e}{\lenT} \SpaSize \log \left( \lenT\cdot\lenL\right),
		\end{equation}
		where $e$ is the base of the natural logarithm. For finite and large $\lenT$, the probability of the \ac{mse} being larger than  $D_{\lenT}(\lenL) $ is at most $\lenT^{-\varepsilon}$.
	\end{prop}
	\begin{proof}
		The proof directly follows using similar arguments as in \cite{chan2014non}, where instead of $\lenT$ possible codewords, in the \ac{sqrss} system there are $\lenT\cdot \lenL$ possible codewords.
	\end{proof}
	
	Comparing \eqref{eqn:ComaRate} and Corollary \ref{cor:BoundRate}, which states an upper bound on the corresponding achievable quantization rate when using the \ac{ml} decoder $\Rate_{\varepsilon}(\lenL)$, indicates that $\bar{\Rate}_{\varepsilon}(\lenL) > \Rate_{\varepsilon}(\lenL)$, i.e., the CoMa decoder requires larger quantization rates to achieve the \ac{mse} $D_{\lenT}(\lenL) $  than the \ac{ml} decoder. However, Proposition \ref{pro:ComaRate} implies that, as the sequence length $\lenT$ increases, the asymptotic growth in the number of bits, $\bits = \Rate \lenT$, is $\bits=\mathcal{O}(\SpaSize \log \lenT + \SpaSize \log \lenL)$, i.e., the same as the growth rate characterized in Corollary \ref{cor:asympt} for \ac{sqrss} with the \ac{ml} decoder. Furthermore, Proposition \ref{pro:ComaRate} holds for any sparsity pattern size $\SpaSize$, while the corresponding analysis of the \ac{ml} decoder requires $\SpaSize$ not to grow with the sequence length, i.e.,  $\SpaSize = \mathcal{O}(1)$.
	
	Based on the characterization of the asymptotic growth of the number of bits $\bits$, we obtain the complexity of Algorithm~\ref{CoMAalgo}.
	The computational burden under which \ac{sqrss} is capable of achieving the  \ac{mse} $D_{\lenT}(\lenL) $  when using the CoMa decoder is stated in the following Corollary:
	\begin{corollary}
		\label{cor:ComaComp}
		\ac{sqrss} with the CoMa decoder achieves the \ac{mse} $D_{\lenT}(\lenL) $ in the limit $\lenT \rightarrow \infty$ with complexity on the order of $\mathcal{O}( \lenT\cdot \lenL \cdot \SpaSize \log (\lenT\cdot \lenL))$ operations.
	\end{corollary}
	
	\begin{IEEEproof}
		Algorithm~\ref{CoMAalgo} essentially scans over all the $\lenT \cdot \lenL$ codewords, comparing each to the $\bits$-bits binary $\myY_{\lenT}$. Consequently, its number of operations is on the order of $\mathcal{O}(\lenT \lenL\bits)$. Combining this with the observation that for $\bits=\mathcal{O}(\SpaSize \log \lenT + \SpaSize \log \lenL)$, \ac{sqrss} with the CoMa decoder achieves the \ac{mse}  $D_{\lenT}(\lenL) $  in the limit $\lenT \rightarrow \infty$ proves the corollary.
	\end{IEEEproof}
	
	Comparing the  complexity of the CoMa method in Corollary \ref{cor:ComaComp} to that of the \ac{ml} decoder in Corollary \ref{cor:complexity} reveals the computational gains of Algorithm~\ref{CoMAalgo}. Focusing on highly sparse setups where $\lenT \gg \SpaSize$, and recalling that $\binom{\lenT}{k\SpaSize} > \frac{\lenT^\SpaSize}{\SpaSize^\SpaSize}$, Corollary \ref{cor:complexity} indicates that the complexity of the \ac{ml} decoder is larger than a term which is dominated by $\lenT^\SpaSize \log \lenT$. Consequently, the \ac{ml} decoder becomes infeasible as $\SpaSize$ grows.
	The corresponding computational complexity of the CoMa decoder in Corollary \ref{cor:ComaComp} is dominated by the term $\lenT\log \lenT$, implying that it  can be implemented for practically any sparsity pattern satisfying $\lenT \gg \SpaSize$. In fact, unlike the \ac{ml} decoder, Algorithm~\ref{CoMAalgo} is invariant to the value of $\SpaSize$, which is required here only for setting the distribution of the codewords, as explained in Subsection \ref{subsec:codebook}, and for determining the quantization rate under which a desired \ac{mse} level is achievable with sufficiently high probability.

	Finally, we note that the CoMa method detailed here is only one example of an efficient algorithm given in the literature that can be leveraged by the \ac{sqrss} system decoder. In fact, it is possible to use several efficient algorithms proposed initially designed for the purpose of traditional group testing, see e.g., \cite{aldridge2014group,coja2019information},  or even modify the coding scheme to be based on systematic codes, as studied in \cite{bui2019efficient}. \rev{An additional decoding scheme which can be combined with SQuaTS is the recently proposed two-stage multi-level decoding method recently proposed in  \cite{cohen2020multi} for COVID-19 pooled testing.} We leave the analysis of \ac{sqrss} with these alternative decoding methods to future investigation.
%
%
	\fi 
	
	\section{\ac{sqrss} for Distributed Quantization}\label{sec:DistQuant}
	In distributed quantization, a set of signals are acquired individually and jointly recovered. Such setups correspond, e.g., to sensor arrays, where each bit-constrained sensor observes a different time sequence, and all the measured sequences should be recovered by some centralized server. Each sensor may have a direct link to the server, resulting in a single hop structure, or must convey its quantized measurements over a route with intermediate nodes, representing a multi-hop topology.
	In this section we show how \ac{sqrss} can be naturally applied to distributed quantization setups.
	 We first formulate the system model for distributed quantization in Subsection \ref{subsec:DistQuantModel}. Then in Subsection \ref{subsec:DistQuantSingle} we adapt \ac{sqrss} to distributed quantization over single-hop networks, and discuss how to it can applied to multi-hop networks in Subsection \ref{subsec:DistQuantMulti}.
	
	\subsection{Distributed Quantization System Model}\label{subsec:DistQuantModel}
	We consider distributed acquisition and centralized reconstruction of $\lenX$ analog time sequences. The sequences, denoted $\{s_m[i]\}_{m=1}^{\lenX}$  are separately observed over the period $i \in  \mySet{\lenT}$, representing, e.g., sources measured at distinct physical locations. The signals are jointly sparse with joint support size $\SpaSize$ \cite{baron2009distributed}. We focus on two models for the joint sparsity of $\{s_m[i]\}$:
	
	\paragraph*{Overall sparsity} Here, the ensemble of all $\lenX$ signals over the observed duration is $\SpaSize$-sparse, namely, the set $\{s_m[i]\}_{m \in \lenXset, i\in\mySet{T}}$\rev{, with $\lenXset\triangleq \{1,\ldots,\lenX\}$,} contains at most $\SpaSize$ non-zero entries. This model, in which no structure is assumed on the sparsity pattern of each signal, coincides with the general joint-sparse model of \cite{baron2009distributed} without a shared component.
	\paragraph*{Structured sparsity} In the second model the signals are sparse in both time and space. Specifically, for each $m \in \lenXset$, the signal $\{s_m[i]\}_{ i\in\mySet{T}}$ is $k_t$-sparse, while for any $i\in \mySet{T}$, the set $\{s_m[i]\}_{m \in \lenXset}$ is $k_s$-sparse. This setup is a special case of overall sparsity with $\SpaSize = k_s k_t$ with an additional structure which can facilitate recovery.

	Each time sequence $\{s_m[i]\}_{i\in\mySet{\lenT}}$ is encoded  into a $\bits$-bits vector denoted $\myX_m\in \{0,1\}^{\bits}$.
	The encoding stage is carried out in a distributed manner, namely, each  $\myX_m$ is determined only by its corresponding time sequence $\{s_m[i]\}_{i\in\mySet{\lenT}}$ and is not affected by the remaining sequences. The binary vectors $\{\myX_m\}$ are conveyed to a single centralized decoder over a network, possibly undergoing several links over multi-hop routes. We consider a binary network model, such that each link can be either broken or error-free. The centralized decoder maintains links with $\lenZ$ nodes. The $\lenZ$ network outputs, denoted $\{\myZ_m\}_{m=1}^{\lenZ}$, are collected by the decoder into a $\bits$-bits vector $\myY\in\{0,1\}^{\bits}$, which is decoded into a digital representation of $\{{s}_m[i]\}$, denoted $\{\hat{s}_m[i]\}$, as illustrated in Fig. \ref{fig:DistQCS_System}.
	The accuracy  is measured by the \ac{mse}   $\sum_i \sum_m\mathbb{E}\left[ (s_m[i]  \!-\! \hat{s}_m[i] )^2\right]$ and the quantization rate  is $\Rate = \frac{\bits}{\lenX\lenT}$ here.
	
	\begin{figure}
		\centering
		{\includegraphics[trim=0cm 0.0cm 0cm 0cm, width = \columnwidth]{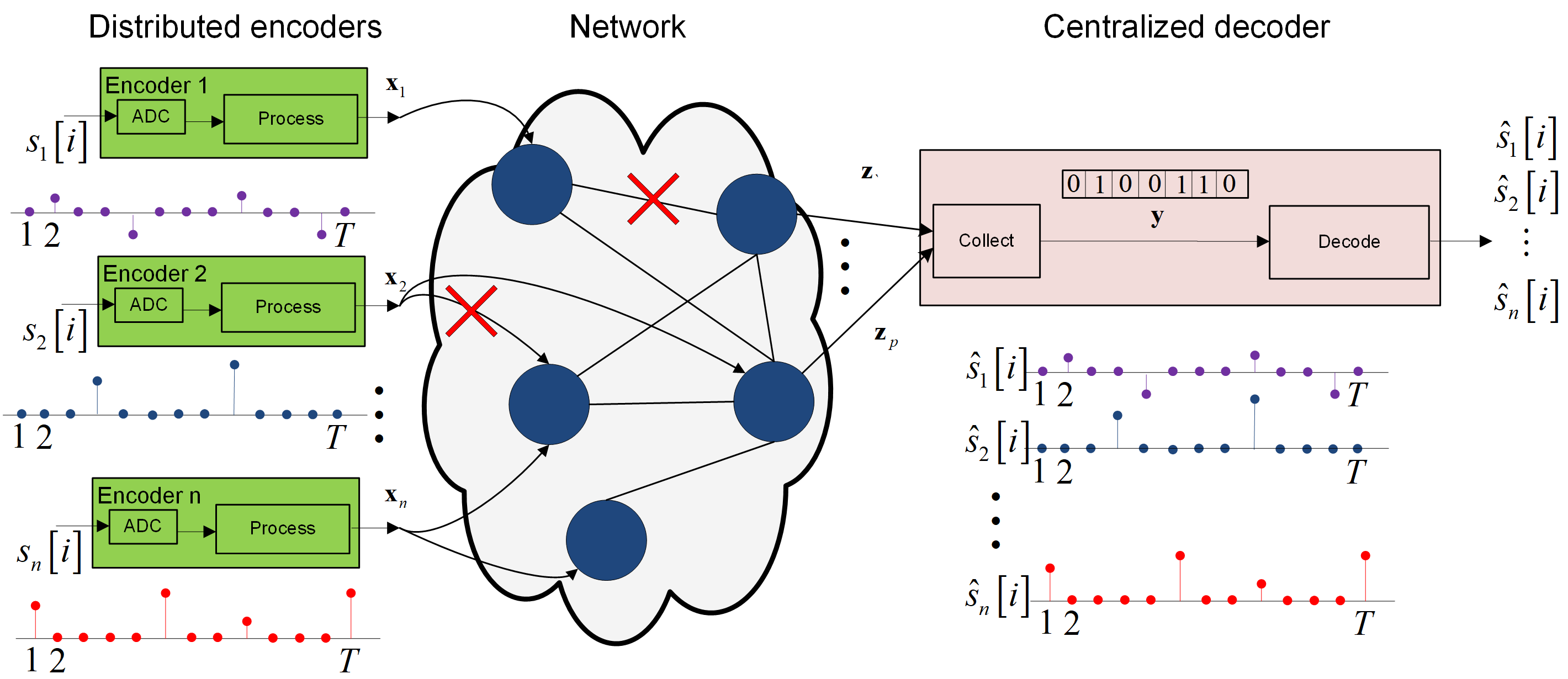}}
		\caption{Distributed quantization system illustration.}
		\label{fig:DistQCS_System}
	\end{figure}
	
%
%
%
		
	\subsection{Single-Hop Networks}\label{subsec:DistQuantSingle}
	We next show how \ac{sqrss}, proposed in Section \ref{sec:sqrss} for the quantization of a single time sequence, can be adapted to distributed quantization setups. We begin here with single-hop networks, where  each encoder has a direct error-free link to the centralized decoder. In particular, the applicability of \ac{sqrss} to distributed quantization stems from the fact that its encoding procedure, and specifically step \ref{itm:E3}, is based on applying a logical OR operation to a set of codewords, selected according to the quantized values of each observed sample.  The associative property of the logical OR operation implies that an \ac{sqrss} encoder can be applied to an ensemble of $\lenX$ sequences by separately encoding each sequence with an \ac{sqrss} encoder using a different codebook.
	
	In particular, in order to apply \ac{sqrss} in a distributed quantization method, one must simply generate a codebook for the ensemble of sequences, i.e., a total of $\lenX \lenT$ codewords, via the generation procedure detailed in Subsection \ref{subsec:codebook}. Then, the generated codewords are distributed among the $\lenX$ encoders, and each encoder of index $m \in \{1,\ldots,\lenX\} $ uses these codewords to apply the \ac{sqrss} encoding method detailed in Subsection \ref{subsec:encoder} to its corresponding sequence $\{s_m[i]\}$. \rev{The output of the $m$-th encoder at time instance $\lenT$, i.e., after its sequence is acquired, is used as the binary vector $\myX_m$ conveyed to the centralized receiver.
	Each node in the network performs a Boolean OR operation of all incoming input vectors $\myX_{m^{\prime}}, m^{\prime} \in \{1,\ldots,P^{\prime}\}$. The output of each node given $P^{\prime}$ inputs vectors is thus $ \vee_{m=1}^{P^{\prime}}\textbf{x}_m$. We denote the output vector of the last nodes before the centralized decoder by $\textbf{z}_m$.} For a single-hop network, the bit vectors received by the decoder, $\{\myZ_m\}_{m=1}^{\lenZ}$, are given by $\myZ_m = \myX_m$ and $\lenZ = \lenX$. Consequently, by the associativity of the logical operator, the centralized decoder can recover the output of applying an \ac{sqrss} encoder to {\em the ensemble of sequences}, denoted $\myY_{\lenT}$, from the outputs of the separate encoders, via
	\begin{equation}
	\label{eqn:recovery}
	\myY_{\lenT}  = \bigvee_{m=1}^{P}\myZ_m.
	\end{equation}
	Using $\myY_{\lenT}$, the centralized decoder can recover the estimate of the ensemble of time sequences $\{\hat{s}_m[i]\}$, via conventional \ac{sqrss} decoding, i.e., \ac{ml} decoding detailed in Subsection \ref{subsec:decoder} or the CoMa method presented in Section \ref{sec:efficient_algorithms}.
	
	 The fact that the proposed distributed adaptation of \ac{sqrss} effectively implements the application of \ac{sqrss} to the ensemble of sequences implies that the achievable performance guarantees of \ac{sqrss}, derived in Subsections \ref{subsec:MSE} and \ref{subsec:ComaAnalysis} for the \ac{ml} and CoMa decoders, respectively,  hold also in the distributed setup. For example, by letting $D_{\lenT, \lenX}(l)$ be the \ac{mse} achieved when $\hat{s}_m[i] = Q(s_m[i])$, namely, a desirable \ac{mse} determined only by the \ac{adc} resolution, it follows from Theorem \ref{direct theorem} that $D_{\lenT, \lenX}(l)$ is achievable by the distributed quantization scheme when its rate satisfies the condition stated in the following proposition:
	 \begin{prop}\label{distributed theorem}
	 	\ac{sqrss} adapted to distributed quantization using the \ac{ml} decoder  achieves the  \ac{mse}   $D_{\lenT, \lenX}(l)$ in the limit  $\lenX\lenT \rightarrow \infty$ with $\SpaSize=\mathcal{O}(1)$  when the quantization rate $\Rate$ satisfies the following inequality 	for some $\varepsilon>0$:
	 	%
	 	\begin{eqnarray}\label{eq:distributed theorem}
	 	\Rate \ge \max_{u \in\mySet{I}(\SpaSize)  }\frac{(1+\varepsilon)k}{u \cdot \lenX T}\log\left( \vartheta\cdot \lenL^u\right).
	 	\end{eqnarray}
	 	The parameter $\vartheta$ and the set $\mySet{I}(\SpaSize)$ depend on the type of joint sparsity: for overall sparsity, $\vartheta = \binom{\lenX \lenT}{\SpaSize}$ and $\mySet{I}(\SpaSize) = \{1,\ldots,\SpaSize\}$, while for structured sparsity $\vartheta =  {\lenX \choose \SpaSize_t}{\lenT \choose \SpaSize_s}$, and $\mySet{I}(\SpaSize= k_s k_t) = \{u_t u_s: 1 \leq u_t \leq \SpaSize_t, 1 \leq u_s \leq \SpaSize_s\}$.
	 \end{prop}

	 \begin{IEEEproof}
	 	The proposition follows by repeating the proof of Theorem \ref{direct theorem} given in the Appendix, while setting the length of the sequence to be $\lenX \lenT$, i.e., the length of the ensemble of signals, instead of $\lenT$, and noting that the joint sparsity affects the number of possible codeword combinations. In particular,  $\vartheta$ is the number of possible sets of non-zero entries in the ensemble of signals over which the \ac{ml} decoder searches, used in Lemma \ref{direct lemma1} in the Appendix.
 	 \end{IEEEproof}

 	Proposition \ref{distributed theorem} indicates that, as expected, \ac{sqrss} can exploit structures in the joint sparse nature of the observed signals to improve performance, namely, to utilize less bits while guaranteeing that a desired \ac{mse}  $D_{\lenT, \lenX}(l)$  is achievable.

	\subsection{Multi-Hop Networks}\label{subsec:DistQuantMulti}
		We now generalize our scheme to a multi-hop network, in which multiple directed links relate the distributed encoders and the centralized decoder. The intermediate nodes in the networks, which act as helpers or relays, can perform basic operations on their input from incoming links. For the sake of space and exposition, we consider a simplified model for this communication network, in which links are assumed to support $b$-bits of information without errors, or result in a complete erasure. We also assume that the transmission is synchronized, i.e., the encoders and intermediate nodes all transmit in sync across their outgoing links, and that the network is acyclic. Note that despite its simplicity, this model is reminiscent of several network models used in the literature, e.g., \cite{el2011network}. 
		
		The operation of the encoders and the decoder in the multi hop setup is identical to that discussed for single hop networks in Subsection \ref{subsec:DistQuantSingle}. The only addition is in the network policy, as depicted in Fig. \ref{fig:QCS_MHP}: At each intermediate node, we perform a Boolean OR operation of all incoming input vectors (which is the same mathematical operation performed by the encoder and decoders in Subsection~\ref{subsec:DistQuantSingle}), and transmit the result length $b$-vector on all outgoing links. The network outputs are collected in $\myY_{\lenT}$ via \eqref{eqn:recovery}.
		
		\begin{figure}
			\centering
			{\includegraphics[trim=0cm 0.0cm 0cm 0cm, width = 0.95\columnwidth]{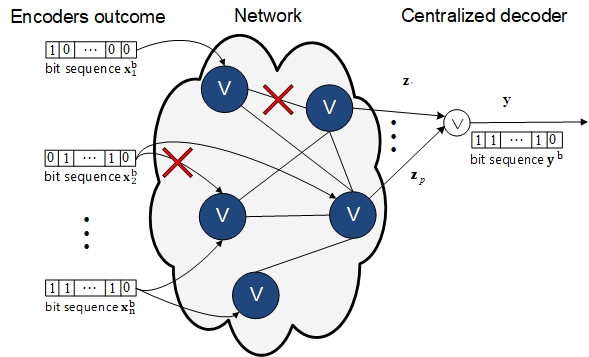}}
			\caption{Acquisition process in multi-hop  Networks.}
			\label{fig:QCS_MHP}
		\end{figure}
		
		Clearly, the resulting bit sequence at the decoder $\myY_{\lenT}$ is identical to the one in  Subsection~\ref{subsec:DistQuantSingle}, as long as there exist at least one path in the network from each encoder to the centralized decoder. Note that this is in  contrast with the previous literature on distributed \ac{cs} over networks, where it is typical to impose  conditions on the network topology that guarantee a successful description \cite{feizi2010compressive}.
		Consequently, the structures of the encoders and the decoder are invariant to whether the encoders communicate with the decoder directly or over multi-hop networks, and the achievable performance of \ac{sqrss} stated in, e.g., Proposition \ref{distributed theorem}, hold in such multi-hop networks.
		Additionally, the scheme we propose is robust to link failures:  as long as there exist at least one path from all encoders to the decoder, any number of link failures in the network still leads to the same received vector at the decoder, i.e., the coding scheme can achieve the min-cut max-flow bound of the network \cite{el2011network,dantzig2003max}.

		It follows from the above discussion that the presence of a multi-hop network does not affect the  operation of the distributed adaptation of \ac{sqrss} or its achievable performance, and only requires a simplified network policy to be carried out by the intermediate network nodes. While our analysis assumes that each encoder has at least a single path to the decoder, it can be shown that the presence of missing paths for some encoders does not affect the recovery of the remaining signals. In particular, by treating the output of a broken link as the zero vector, if the $m$th encoder has no path to the decoder, the recovery of $\{s_j[i]\}_{j\neq m}$ remains intact, while $\hat{s}_m[i]$ is estimated as being all zeros.

%
%
%
		
	\section{Numerical Evaluations}\label{sec:sims}
	In this section, we evaluate the performance of the proposed \ac{sqrss} scheme using various simulations, for a fixed and finite signal size $\lenT$. We first numerically evaluate \ac{sqrss} for the quantization of a single sparse signal in Subsection \ref{subsec:SimSingle}. Then, in Subsection \ref{subsec:SimNoise} we study the resiliency of \ac{sqrss} to noisy digital representations. Finally, in Subsection \ref{subsec:SimDist}, we evaluate its extension to distributed quantization setups, as discussed in Section \ref{sec:DistQuant}.

    \subsection{Single Sparse Signal}
    \label{subsec:SimSingle}
    We begin by numerically evaluating \ac{sqrss} used for quantizing a single sparse signal $\{s[i]\}$, $i \in \mySet{\lenT}$.
	To this end, we consider two sparse sources with sizes $\lenT \in \{100,50\}$ and support sizes $\SpaSize \in \{3,2\}$, respectively. To generate each signal, we  randomly select $\SpaSize$ indexes, denoted $\{i_j\}_{j=1}^\SpaSize$, and then choose the values of $\{s[i_j]\}$ to be i.i.d. zero-mean unit-variance Gaussian \acp{rv}, while the remaining entries   are set to zero. 
	
	Each of the generated signals is quantized and represented in digital form using each of the following methods:
	\begin{itemize}
		\item \ac{sqrss} system with   $\lenL = \rev{\max\left(\lfloor\frac{1}{\lenT}2^{\frac{\lenT\Rate}{\SpaSize(1+\epsilon)}}\rfloor , 2\right)}$ following \eqref{eqn:BoundRate}, where $\epsilon$ is selected \rev{empirically for each point} in the range $\epsilon \in [0.8, 1.3]$ \rev{to maximize the performance of the \ac{ml} decoder, where we select the value which achieves the minimal MSE among four different values of $\epsilon$ uniformly placed in this region}. Here, the continuous-to-discrete mapping $Q(\cdot)$ implements uniform quantization over the region $[-2,2]$.
		\ifFullVersion
		We consider \rev{three different \ac{sqrss} decoders: the \ac{ml} decoder detailed in Subsection \ref{subsec:decoder} the CoMa-based reduced complexity decoder proposed in Section~\ref{sec:efficient_algorithms} which is tuned with the same value of $\epsilon$ as that used by the \ac{ml} decoder; and the CoMa decoder with $\epsilon$ optimized by fine search separately for each quantization rate.}
		\fi 
		\item A uniform scalar quantizer with support $[-2,2]$ applied separately to each sample of $\{s[i]\} $, mapping every decision region to its centroid. This system, which models the direct application of a serial scalar \ac{adc} to the sparse signal   $\{s[i]\}$, can be utilized only when the quantization rate satisfies $\Rate \ge 1$, as the quantizers require  at least one bit.
		\item A compress-and-quantize system which first compresses  $\myS = [s[1], \dots, s[\lenT]]^T$ into $\mySet{R}^{m}$, where $m$ is selected in the range $[6\SpaSize, 20\SpaSize]$ to minimize the \ac{mse}. The compression is carried out using a sensing matrix $\myMat{A} \in \mySet{R}^{m \times \lenT}$ whose entires are i.i.d. zero-mean unit variance Gaussian \acp{rv}. The compressed signal $\myMat{A}\myS $ is quantized using a uniform scalar quantizer with support $[-2,2]$. The digital representation $\hat{\myS} = [\hat{s}[1], \dots, \hat{s}[\lenT]]^T$ is then recovered using the \ac{qiht} method \cite{jacques2013quantized} as well as \ac{fista} \cite{beck2009fast}.
	\end{itemize}
	All of the above schemes are compared with the same number of bits $\bits = \Rate \cdot \lenT$, and the \ac{mse} is computed by averaging the squared error over $100$ Monte Carlo simulations.

	\begin{figure}
		\centering
		\ifFullVersion		
		{\includegraphics[trim=0cm 0.0cm 0cm 0.0cm, width = \columnwidth]{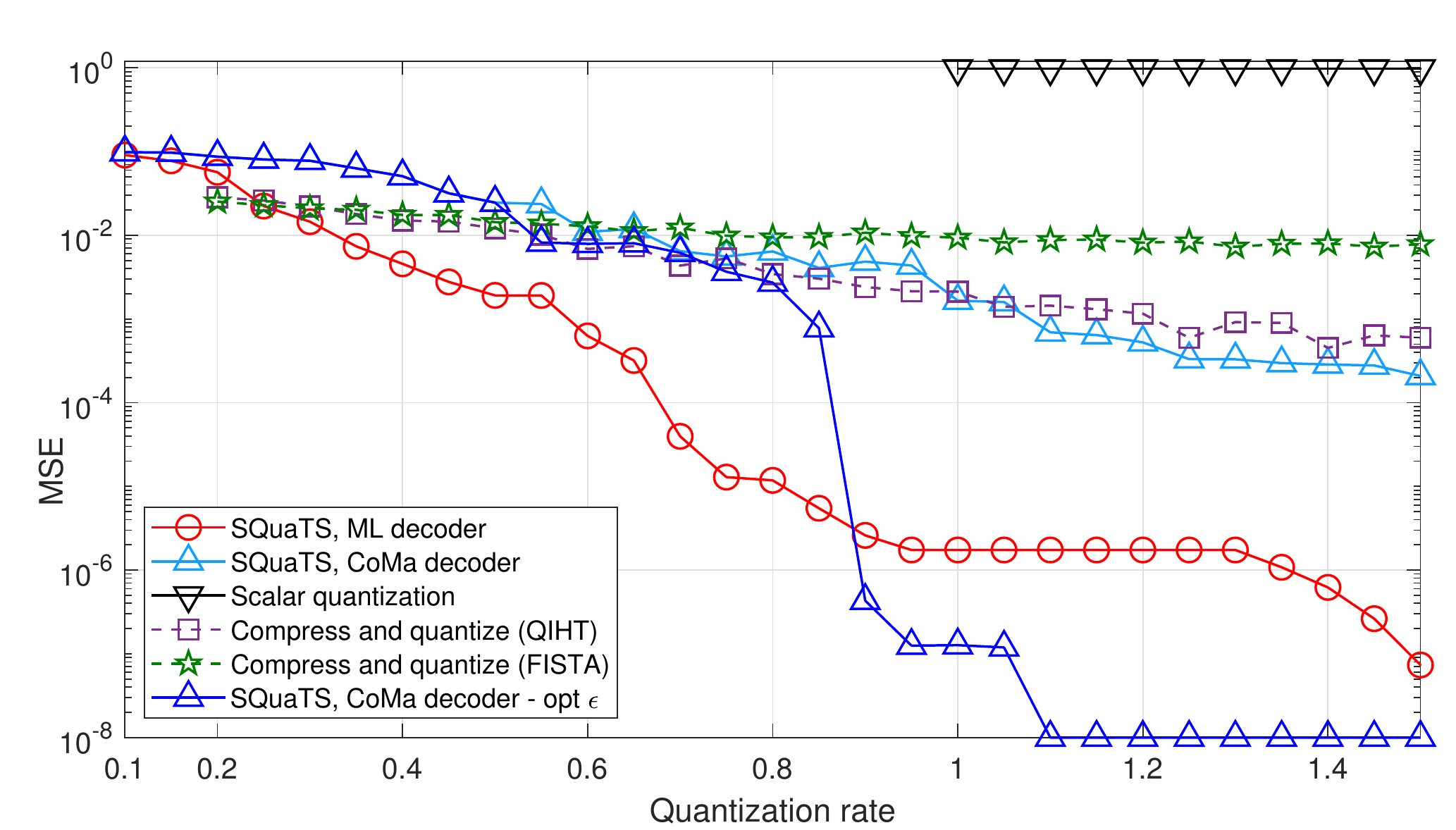}}
		\else
		{\includegraphics[width = \columnwidth]{matlab/Quantization_rate_4_comp7a.eps}}
		\fi 
		\caption{Quantization systems comparison, $\lenT = 100$, $\SpaSize = 3$.}
		\label{fig:simulation}
	\end{figure}
	
	\begin{figure}
		\centering
		\ifFullVersion
		{\includegraphics[trim=0cm 0.0cm 0cm 0.0cm, width = \columnwidth]{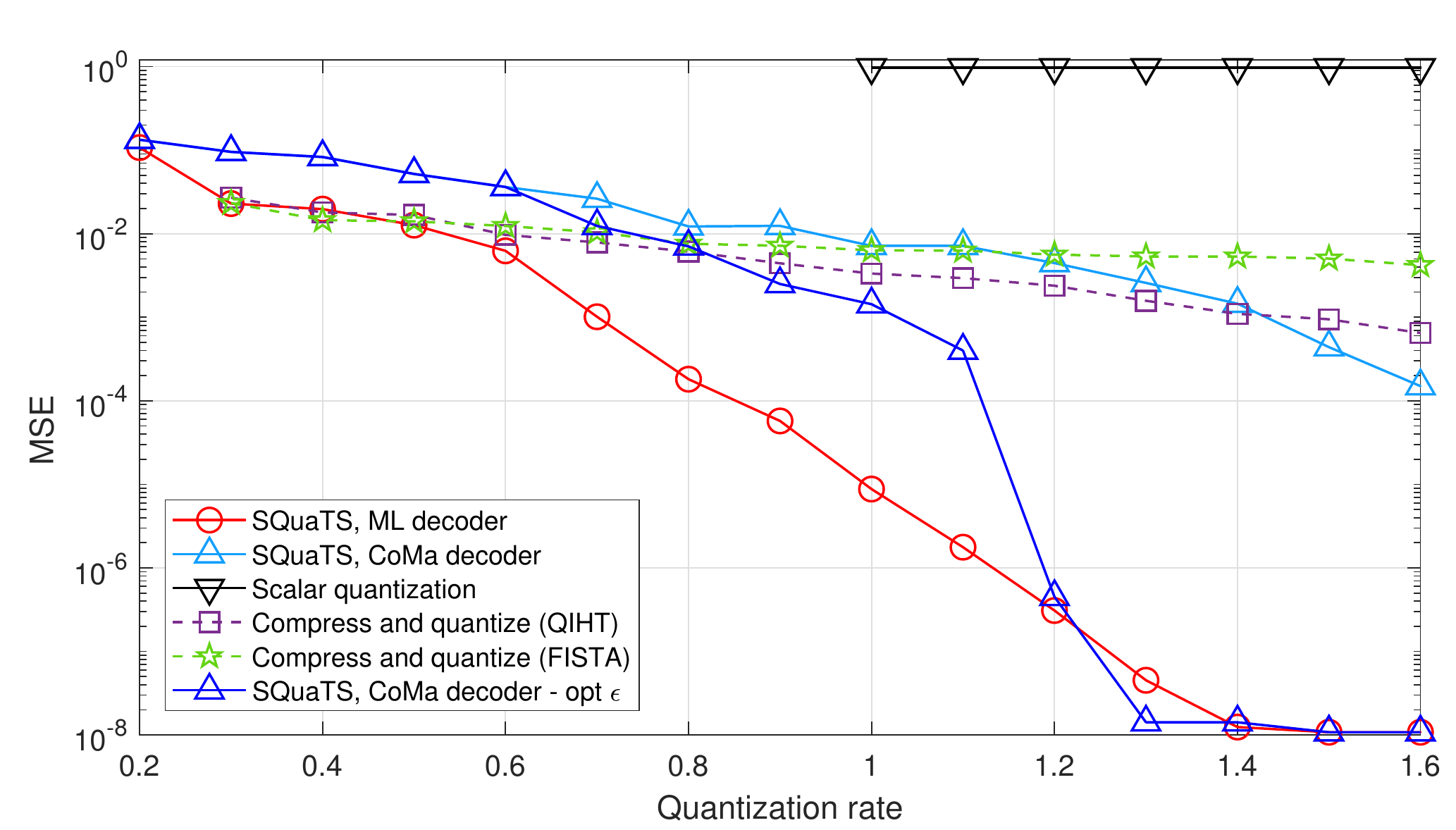}}
		\else		
		{\includegraphics[width = \columnwidth]{matlab/Quantization_rate_1_comp_n50_k2_4_V2.eps}}
		\fi 
		\caption{Quantization systems comparison, $\lenT = 50$, $\SpaSize = 2$.}
		\label{fig:simulation1}	
	\end{figure}
	
	The empirically evaluated \acp{mse} of the considered quantization systems versus the quantization rate $\Rate$ are depicted in Figs. \ref{fig:simulation}-\ref{fig:simulation1} for the setups with $(\lenT, \SpaSize) = (100,3)$ and $(\lenT,\SpaSize) = (50,2)$, respectively. Observing Figs. \ref{fig:simulation}-\ref{fig:simulation1}, it is noted that the proposed \ac{sqrss} system achieves superior representation accuracy and that its resulting \ac{mse} is not larger than $10^{-4}$ for quantization rates $\Rate \ge 0.8$. For comparison, directly applying a scalar quantizer to the sparse signal is feasible only for $\Rate \ge 1$, and its achievable \ac{mse} is only slightly less than $1$. This degraded performance of directly applying scalar quantizers stems from the fact that for the considered rates $\Rate \in [1,2)$, this quantization mapping implements a one-bit sign quantization of $\{s[i]\}$. Since most of the samples of $\{s[i]\}$ are zero, this quantization rule induces substantial distortion.
	
	The \ac{mse} performance of the \ac{cs}-based quantization scheme improves much less dramatically with the quantization rate $\Rate$ compared to   \ac{sqrss}. For example, for the scenario depicted in Fig. \ref{fig:simulation}, the \ac{sqrss} system achieves \ac{mse} of $2\cdot 10^{-3}$ for $\Rate = 0.5$, while the \ac{cs}-based systems achieve an \acp{mse} of $1.2\cdot 10^{-2}$ and $1.4 \cdot 10^{-2}$ for the \ac{qiht} and \ac{fista} decoders, respectively, i.e., a gap of approximately $7$ dB. However, for quantization rate of $\Rate = 1$, the corresponding \ac{mse} values are $1.7 \cdot 10^{-6}$, $2\cdot 10^{-3}$, and  $9\cdot 10^{-3}$, for the \ac{sqrss} system, \ac{cs} with \ac{qiht} recovery, and \ac{cs} with \ac{fista} recovery, respectively, namely, performance gaps of $30 - 37$ dB in \ac{mse}. For all considered scenarios, the \ac{qiht} recovery scheme, which is specifically designed for reconstructing sparse signals from compressed and quantized measurements, outperforms the \ac{fista} method which considers general sparse recovery.
	
	\ifFullVersion
	
	\begin{figure}
	    \centering
	    {\includegraphics[trim=0cm -1.0cm 0cm 0.0cm, clip, width = 1 \columnwidth]{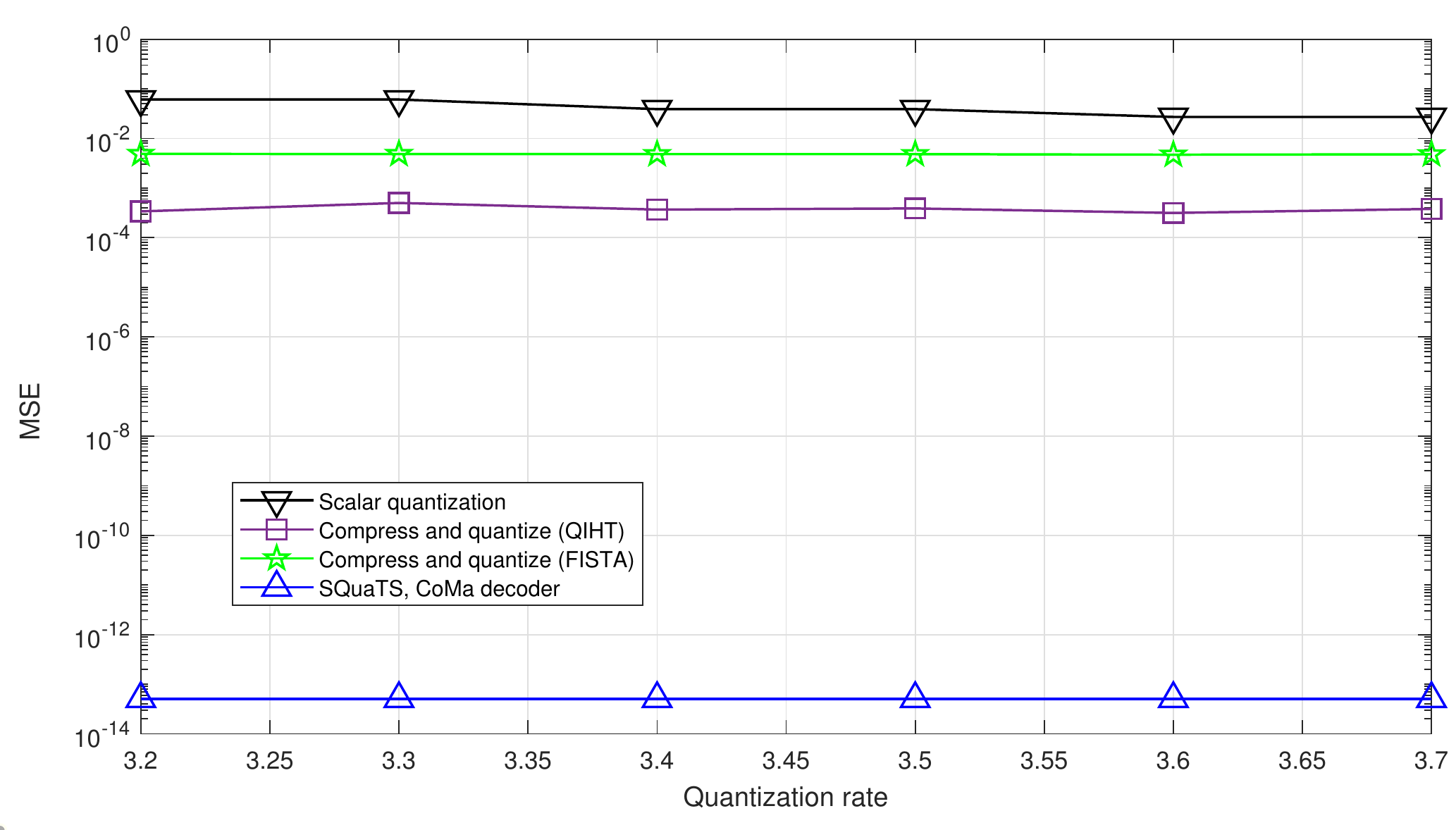}}
	    {\includegraphics[trim=0cm 0.0cm 0cm 0.0cm, clip, width = 1 \columnwidth]{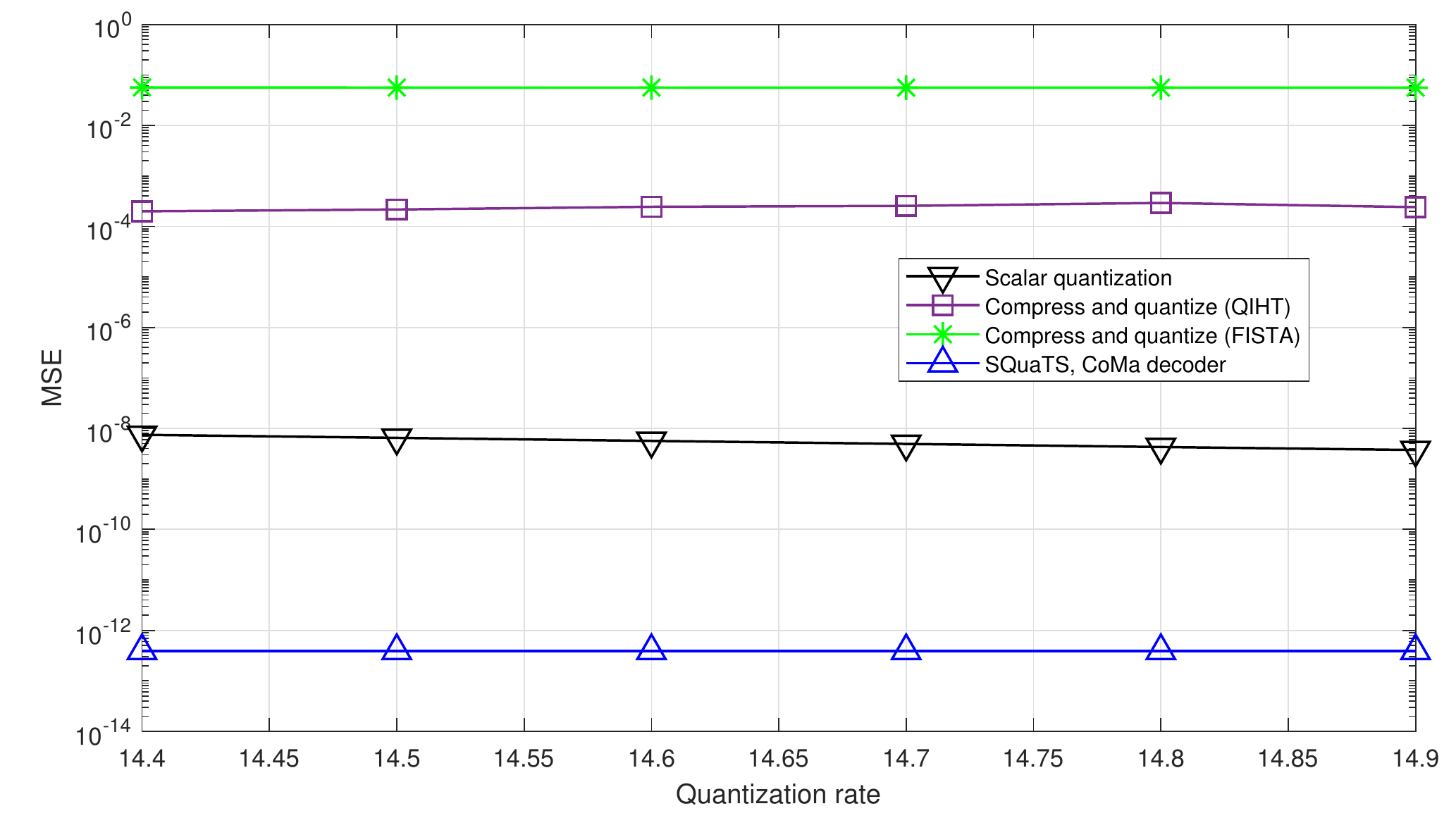}}
	    \caption{\rev{Quantization systems comparison with CoMa decoder for $T = 50$. In the top panel $k=2$ and in the bottom panel $k = 10$.}}
	    \label{fig:Coma_per}
    \end{figure}

	Furthermore, it is observed in Figs. \ref{fig:simulation}-\ref{fig:simulation1} that the \ac{sqrss} encoder combined with the reduced complexity CoMa detector proposed in Section \ref{sec:efficient_algorithms} outperforms \ac{cs}-based quantizers as the quantization rate $\Rate$ increases when using both a fixed $\epsilon$ as well as an optimized setting of this parameter. In particular, for the scenario whose results are depicted in Fig. \ref{fig:simulation}, \ac{sqrss} with the CoMa decoder without $\epsilon$ optimization outperforms compress-and-quantize with \ac{qiht} and \ac{fista} recovery for rates $\Rate > 0.6$ and $\Rate > 1$, respectively. The corresponding rate thresholds for the scenario in Fig. \ref{fig:simulation1} are $\Rate > 1.1$ and $\Rate > 1.4$, respectively. However, when optimizing $\epsilon$ specifically for the CoMa decoder in a fine manner, which can be carried out without substantial overhead due to its reduced computational complexity, its performance is notably improved. In fact, in high quantization rates, we observe that the optimization of $\epsilon$ allows the CoMa decoder to outperform the ML decoder, for which $\epsilon$ is merely selected out of four possible candidates.
	
	\rev{The ability of \ac{sqrss} with the CoMa decoder to notably outperform \ac{cs}-based approaches is also demonstrated for higher resolution quantization regimes in  Fig.~\ref{fig:Coma_per}. This is achieved by increasing the quantization rate to satisfy the sufficiency rate in Proposition~\ref{pro:ComaRate} and selecting appropriate $l$.  Fig.~\ref{fig:non_zero_samples} numerically evaluates \ac{sqrss} combined  with fragmentation under the scenario presented in the bottom panel of Fig.~\ref{fig:Coma_per}. In particular, fragmentation is carried out here to reduce the computational burden by dividing the signal into ten groups, as discussed in Section~\ref{subsec:MSE}, resulting in $T=5$ and $k=1$. We observe that CoMa  achieves improved performance also for various quantization rates by properly setting $l$, as demonstrated in Fig.~\ref{fig:high_k}.} These results indicate that \ac{sqrss}, which acquires the sparse signal in a serial manner, is capable of achieving superior performance even when using sub-optimal reduced complexity decoders, compared to \ac{cs}-based methods which must observe the complete signal before it is quantized.
	\fi 

	\begin{figure}
	    \centering
	    {\includegraphics[trim=0cm 0.0cm 0cm 0cm, clip, width = 1 \columnwidth]{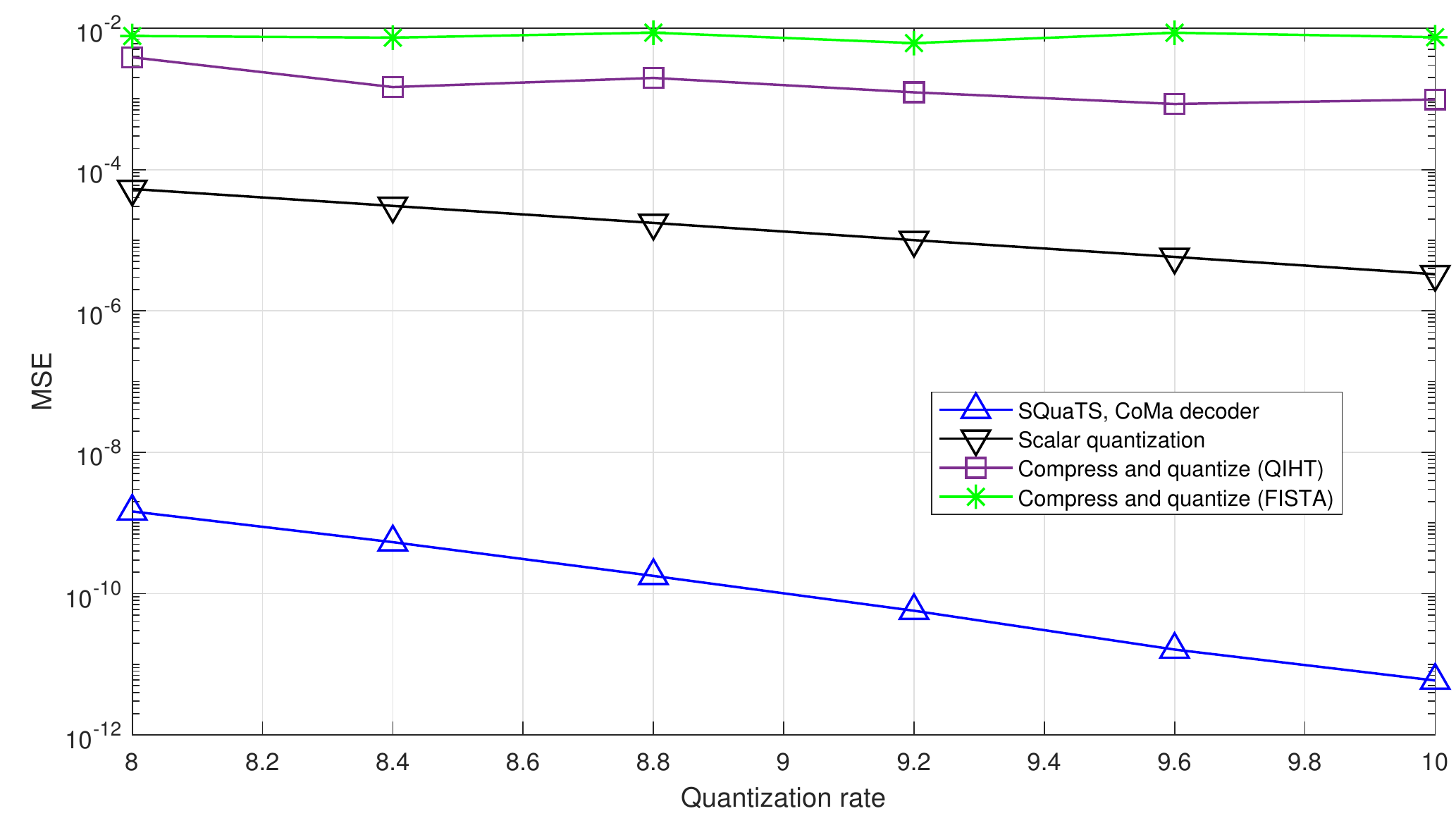}}
	    \caption{\rev{Quantization systems comparison for fragmentation of $T=50$ and $k=10$, as given in Fig. \ref{fig:Coma_per}, to ten groups such that $T = 5$ and $k = 1$.}}
	    \label{fig:non_zero_samples}
    \end{figure}

	\rev{Fig~\ref{fig:diff_k} numerically evaluates the performance obtained for the scenario discussed in Section \ref{subsec:Discussion}, in which the value of $k$ used in designing the \ac{sqrss} system differs from the true sparsity level. Here, we designed the system  for $k=3$, while evaluating its performance when the actual number of the non-zero samples varies in the proximity of $k=3$. For lower non-zero samples case, we observe in Fig. 11 that  \ac{sqrss} still exploits the reduced sparsity level, allowing it to be translated into improved performance, despite the fact that it was designed with higher values of $k$. When the true value of $k$ is larger than that used in design, we observe some graceful degradation in the MSE accuracy of \ac{sqrss}, indicating its ability to maintain reliable operation in such scenarios. In summary, if the number of non-zeros input samples is higher or lower than the a-priori $k$ selected, the codebook designed for $k$ can still be applied. In such cases, when the mismatch is not too large, only a minor degradation in the quantization distortion is observed. However, when this number of non-zeros inputs samples is much different than the a-priori assumption, one will need different codebook designs, and a look up table if this is used to reduce the decoding complexity with ML algorithm.}
    \begin{figure}
	    \centering
	    {\includegraphics[trim=0cm 0.0cm 0cm 0cm, clip, width = 1 \columnwidth]{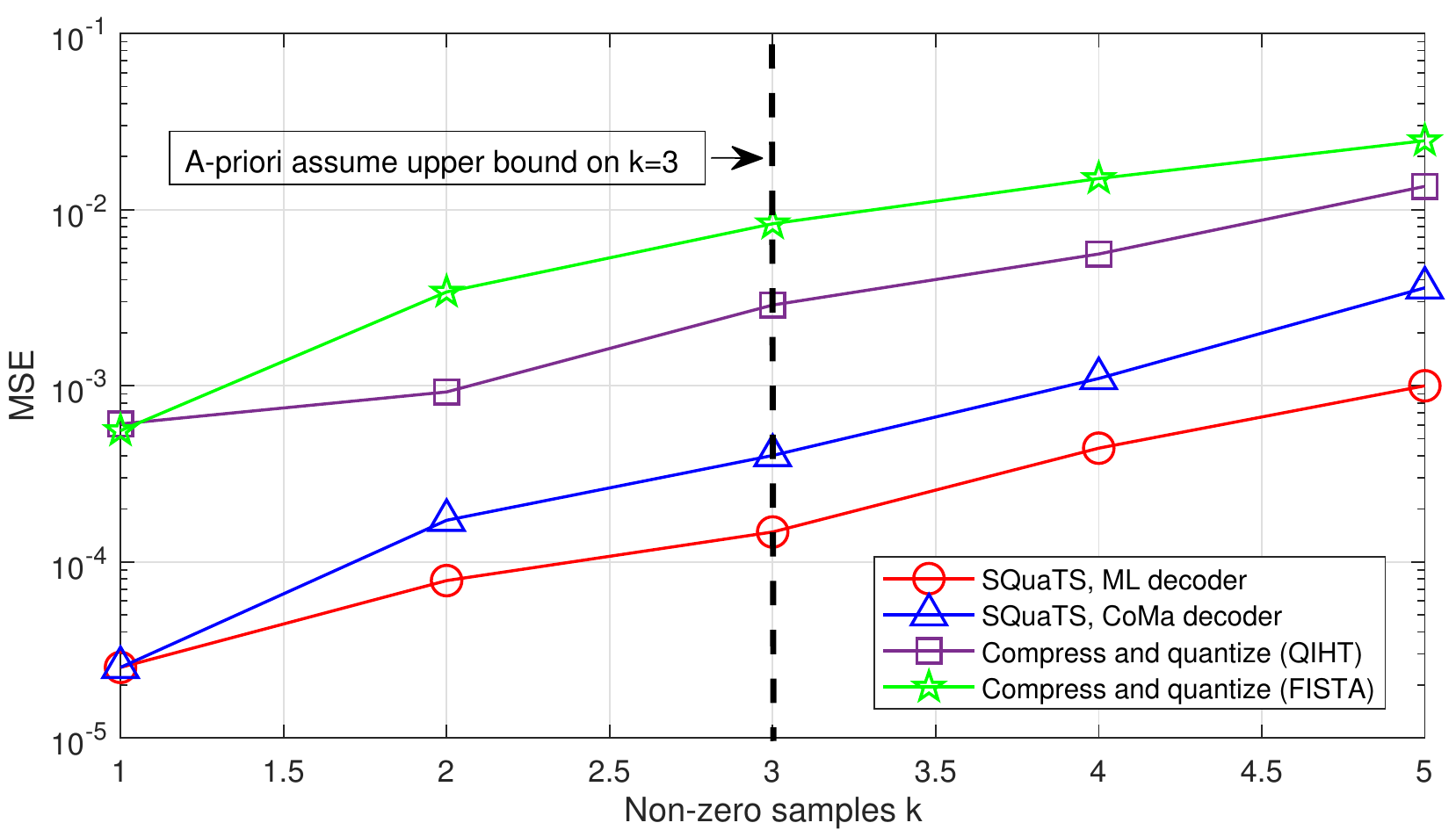}}
	    \caption{\rev{Quantization systems comparison for different numbers of non-zero inputs, $T = 100$, quantization rate of $R=0.78$ and SQuaTS codebook designed for upper bound of $k = 3$.}}
	    \label{fig:diff_k}
    \end{figure}
	
	\off{\rev{In \ac{sqrss}, as we elaborate in Section~\ref{subsec:MSE}, the sparsity is defined as the ratio between $k$ and $T\cdot l$. Hence, for sufficiently large $l$, the samples are considered as sparse even when $k/T \rightarrow 1$. Fig~\ref{fig:non_zero_samples}, present an example for $T=50$ and $k=30$. Where $k=T$, CoMa decoder gets approximately similar MSE results as with scalar quantization.}}
	
	Finally, Figs.~\ref{fig:NumericalResults1} and \ref{fig:NumericalResults2} show the trade-off between the quantization rate required for achieving a desired distortion level $D_{\lenT}(\lenL)$, computed via \eqref{eq:reduce_hw}, with the parameters $k$ and $l$, for $T=50$ and $T = 100$, respectively. The numerical results in Figs.~\ref{fig:NumericalResults1} and \ref{fig:NumericalResults2} demonstrate how the quantization rate scales with $l$ and $k$. These observed curves settle with the characterization in \Cref{cor:BoundRate}, which noted that the rate scales logarithmically with $l$, and linearly with $k$. The results depicted in Figs.~\ref{fig:NumericalResults1} and \ref{fig:NumericalResults2} demonstrate the ability of \ac{sqrss} to exploit the signal sparsity and to benefit from small values of the support $k$, as the required quantization rate is reduced substantially as $k$ decreases \rev{To demonstrate that the performance gains of \ac{sqrss} persist in such regimes, we numerically evaluate its performance  for a  scenario with $T=50$ and $k=5$. This scenario, for which the empirical performance is depicted in  Fig.~\ref{fig:high_k},  is equivalent to the scenario of $T=100$ and $k=10$, in which fragmentation into two groups is used to reduce the decoding complexity. as suggested in Section~\ref{subsec:MSE}. Note that by performing fragmentation, the quantization rate is divided by the number of groups. Fig~\ref{fig:high_k} also presents the performance with CoMa decoding for a low quantization rate when $\epsilon$ is optimized for this decoder. In practice, for high $k$ using ML decoder, fragmentation is needed due to the computational complexity. In the regime we simulated using CoMA, fragmentation is not needed. We observe in Fig.~\ref{fig:high_k} that the expected gains of \ac{sqrss} are indeed maintained here. }
	
    \begin{figure}
		\centering
		{\includegraphics[width = \columnwidth]{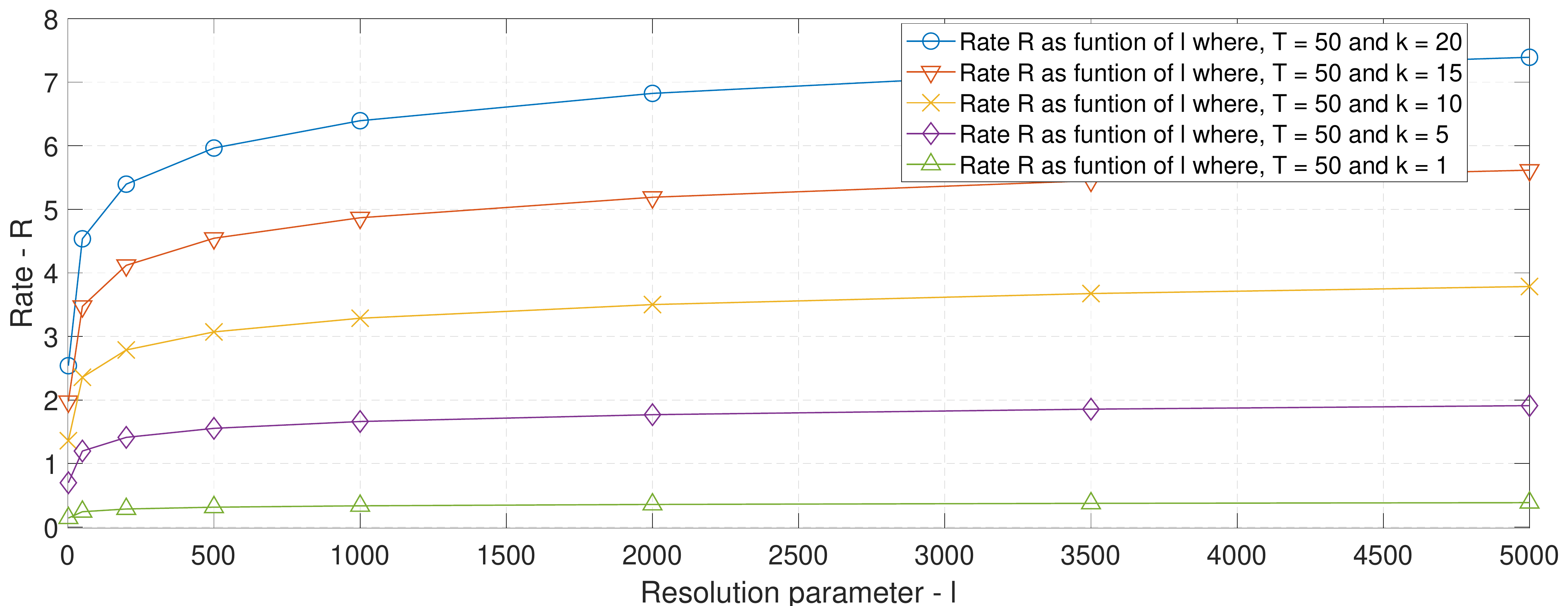}}
		\caption{Numerical results of \ac{sqrss} system, for $T = 50$ and different values of $\lenL$ and $\SpaSize$.
		\label{fig:NumericalResults1}}	
	\end{figure}
	\begin{figure}
		\centering
		{\includegraphics[width = \columnwidth]{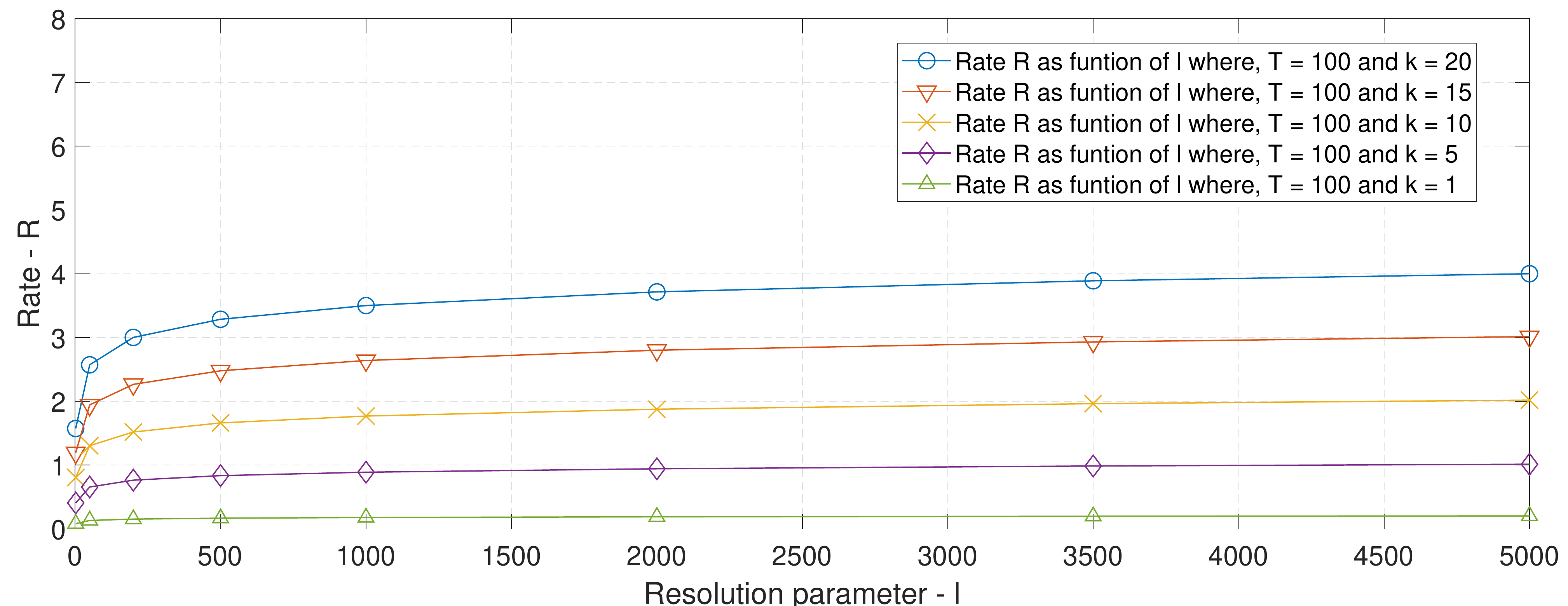}}
		\caption{Numerical results of \ac{sqrss} system, for $T = 100$ and different values of $\lenL$ and $\SpaSize$.}
		\label{fig:NumericalResults2}	
	\end{figure}
	
	\subsection{Noisy Digital Representation}
	\label{subsec:SimNoise}
	
	Next, we carry out a set of experiments whose goal is to demonstrate how \ac{sqrss} can be adapted to deal with noisy digital representations.
	To that aim, we consider the setup in which the value of the bits register  $\myVec{y}_T$ experiences independent random bit-flips, i.e., each bit $y_{i,T}$ is flipped from 0 to 1 with probability $0 \leq q < 1/2$ (positive flip), and flipped from 1 to 0 with probability $0 \leq u < 1$ (negative flip). This noisy model represents corruption of the codeword in the digital domain.
	\off{Note that the fact that this noise models affects the register $\myVec{y}$ at the last time step $T$ implies that it constitutes a \emph{worst-case digital noise}.  This follows since  any noise affecting $\myVec{y}_i$ at time instance $i < T$ is carried through if it corresponds to a positive flip, and can only be corrected if it corresponds to a negative flip\footnote{\textcolor{red}{Alejandro - this sentence is not clear, it implies that adding noisy in time $i <T$ is worse that adding noise at the last time instance, but then you say that adding the noise at $i=T$ is the worst case. So please clarify what you mean.}}, due to the nature of the Boolean operations taking place.}
	
	\begin{figure}
	    \centering
	    {\includegraphics[trim=0cm 0.0cm 0cm 0cm, clip, width = 1 \columnwidth]{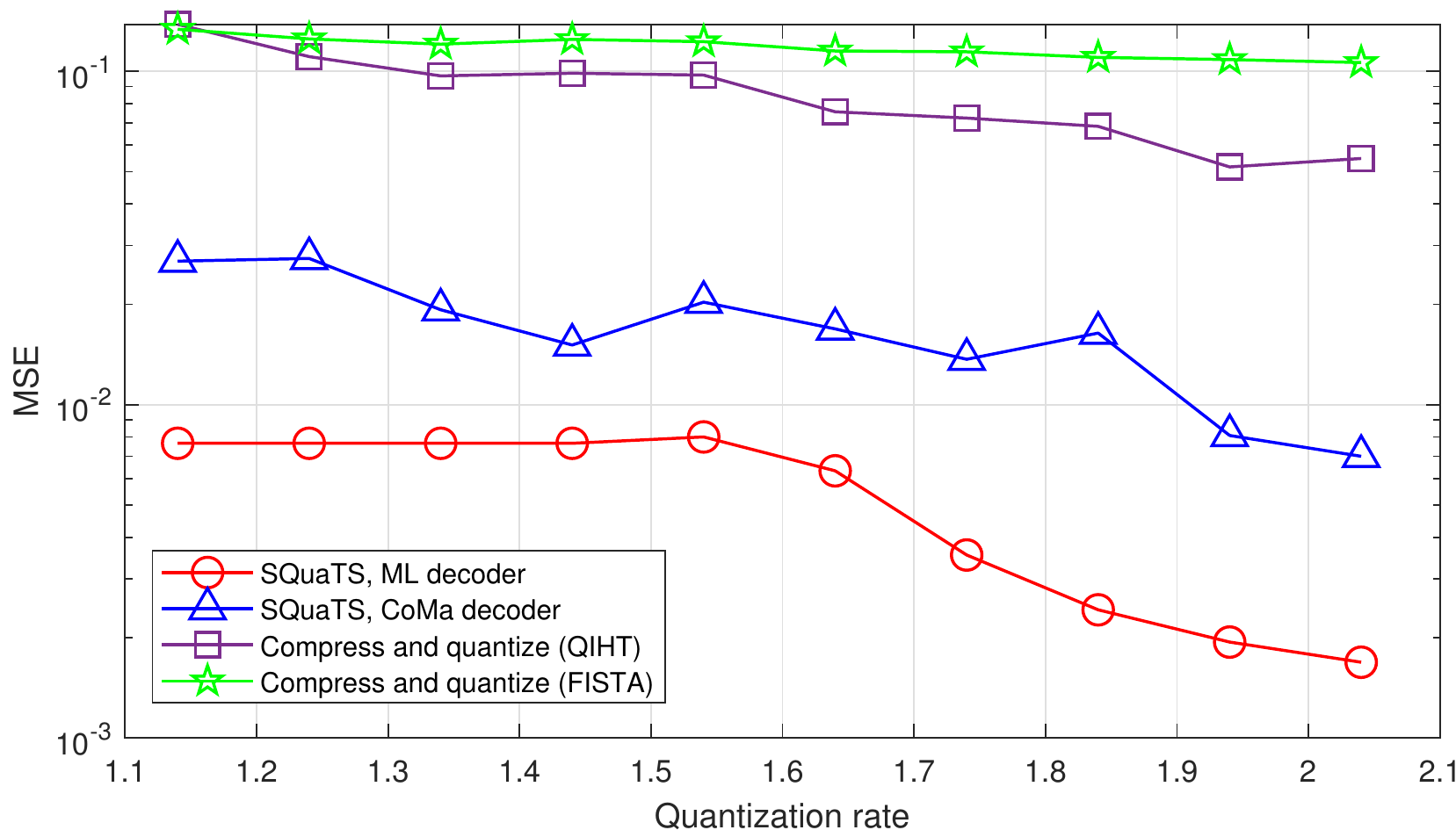}}
	    \caption{\rev{Quantization systems comparison, $T = 50$ and $k = 5$. This scenario is equivalent to the case of $T=100$ and $k=10$, in which we are performing fragmentation to two groups. In this case, the quantization rate is divided by the number of groups.}}
	    \label{fig:high_k}
    \end{figure}
	
	By making use of existing results in noisy group testing theory, we propose to increase the length of the codewords by some factor which depends on $u$ and $q$, while keeping the encoding identical.
	A trivial adaptation of the results in \cite{atia2012boolean} to our setup, reveals that increasing the length of the codeword $b$ by a factor $\frac{1}{1-q} \cdot \frac{1}{(1-u)^2}$ is sufficient\footnote{While this factor is relatively loose, it is preferred here over  the complex yet more precise expression that can be found in \cite{sejdinovic2010note} due to its simple formulation.}.
	On the side of the decoder, the ML scheme operates identically.
	The efficient CoMa method however, must be modified to deal with this noise in the system.
	While a precise discussion of this is outside of the scope of this paper, we refer the interested reader to exisiting efficient algorithms for noisy group testing such as Noisy-CoMa \cite{chan2011non}, as possible ways to tweak the CoMa decoding scheme of Subsection \ref{subsec:Coma} to account for the presence of digital noise.
	
	Figs.~\ref{fig:simulation3} and \ref{fig:simulation4} show the empirical \acs{mse} of the adapted scheme under the same signal settings considered in the previous subsection, in the presence of digital noise with parameters $(q = 0.1, u = 0.1)$ and $(q =0.4, u= 0.1)$, respectively.
	We observe in Figs.~\ref{fig:simulation3}- \ref{fig:simulation4} that the quantization rate $R$ required to achieve a given \ac{mse} level is increased compared to the noiseless case in Figs. \ref{fig:simulation}-\ref{fig:simulation1} -- quantifying the additional bits which enable \ac{sqrss} to be robust to digital noise.
	More precisely, to achieve an average \ac{mse} of $10^{-3}$ with the proposed scheme, the quantization rate must be increased from $.7$, required in the absence of digital noise, to  about $1.6$ and $6$, when $q = 0.1$ and $q=0.4$, respectively.
	As observed, this loss in performance is much more dramatic when $q$ grows, 	revealing that positive flips are more costly than negative flips.
	In either cases, the performance of the adapted \ac{sqrss} scheme outperforms significantly the \ac{cs}-based approaches, which fail to breach under the \ac{mse} of $10^{-2}$, even for large quantization rates.
	
	\begin{figure}
		\centering
		{\includegraphics[width = \columnwidth]{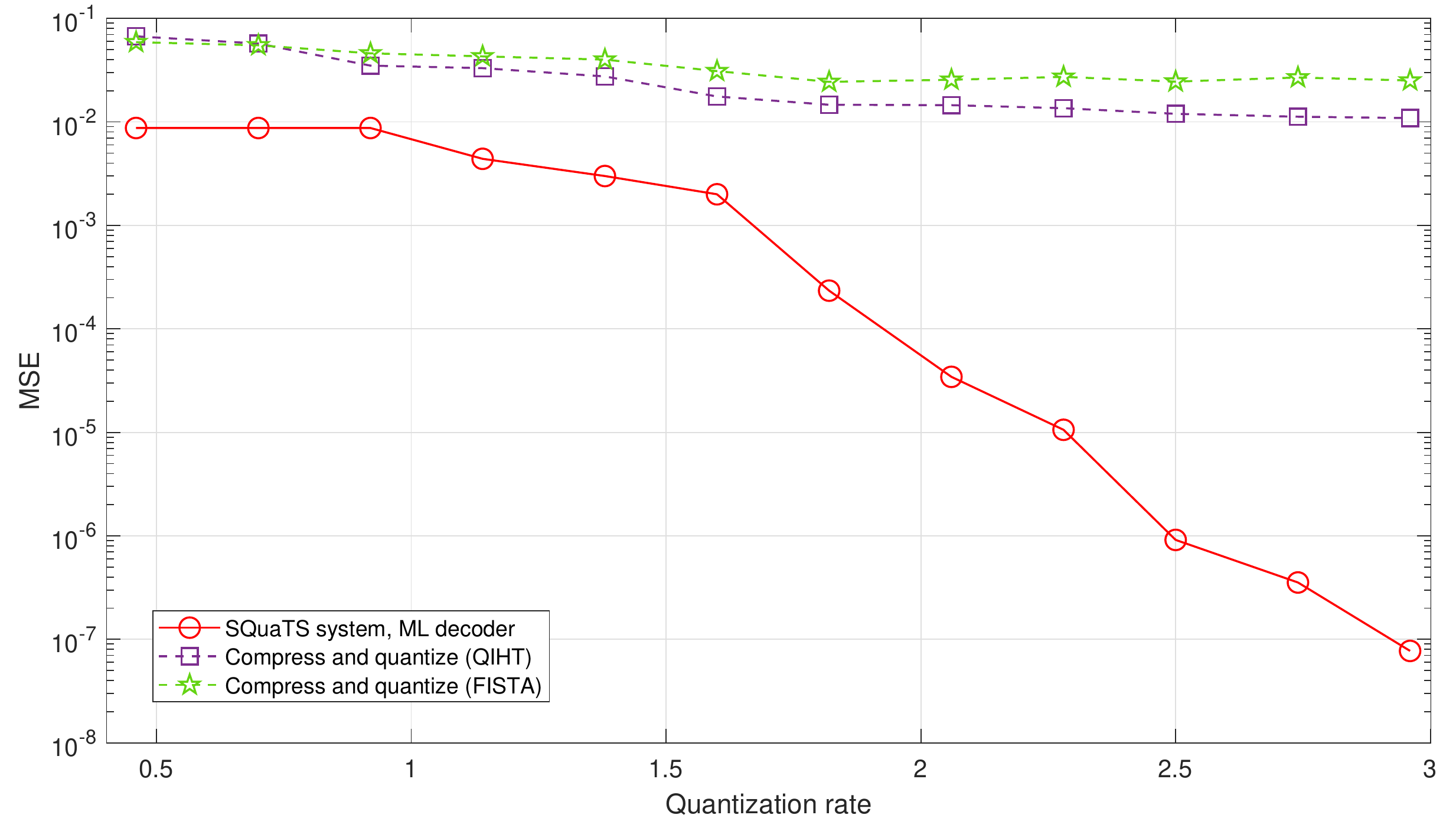}}
		\caption{Quantization systems comparison, $T = 50$, $\SpaSize = 2$, $q=0.1$, $u=0.1$.}
		\label{fig:simulation3}	
	\end{figure}	
	\begin{figure}
		\centering
		{\includegraphics[width = \columnwidth]{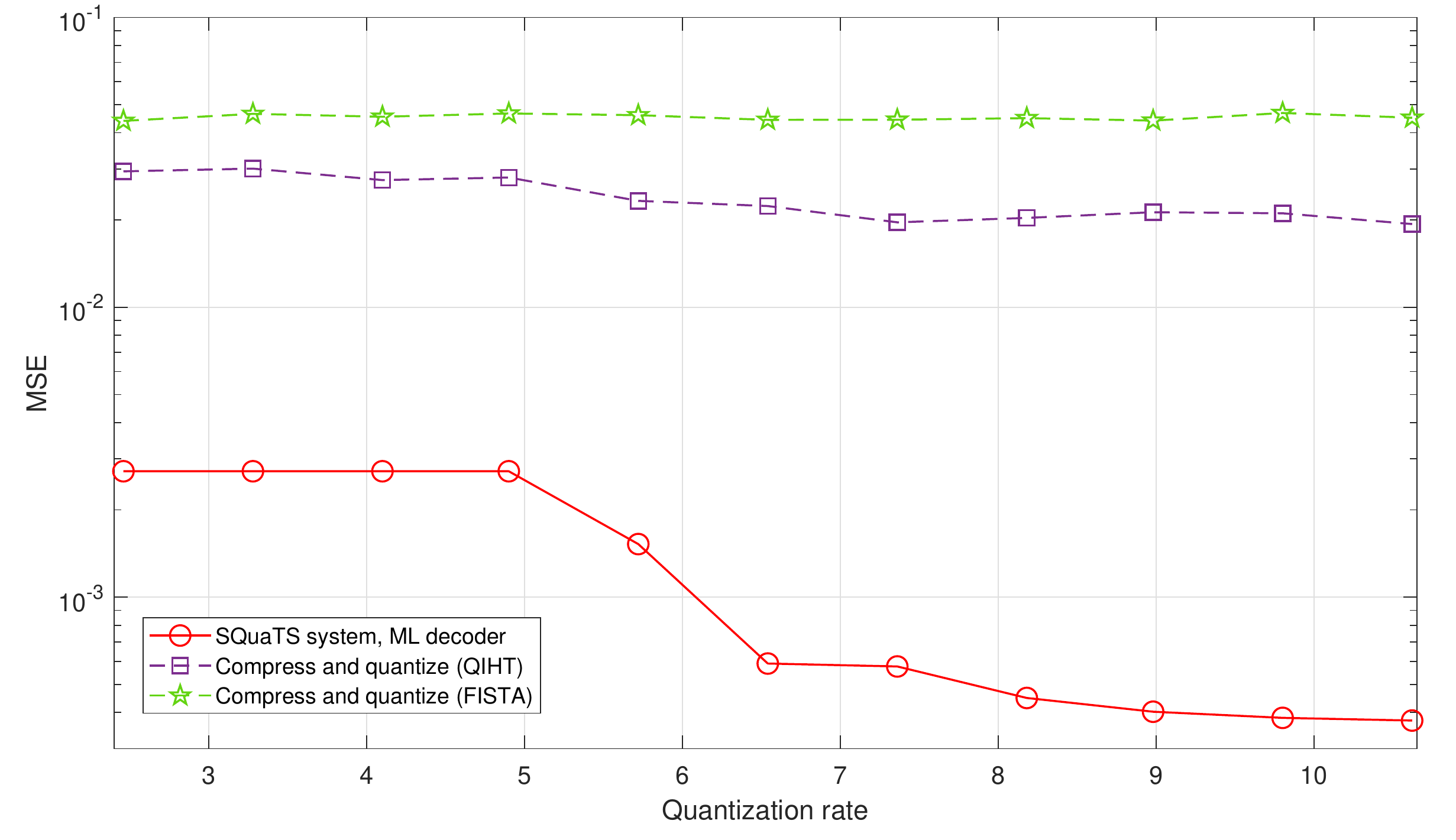}}
		\caption{Quantization systems comparison, $T = 50$, $\SpaSize = 2$, $q=0.4$, $u=0.1$.}
		\label{fig:simulation4}	
	\end{figure}

\subsection{Distributed Quantization}
\label{subsec:SimDist}

We end this section with a numerical study of the extension of \ac{sqrss} to distributed quantization setups, detailed in Section \ref{sec:DistQuant}.
To that aim, we numerically compute the achievable \ac{mse} of \ac{sqrss} applied to $n=10$ jointly sparse signals of size $T=100$ with joint support size of $k=3$ in a single hop network.
In Fig. \ref{fig:distortion} we compare the \ac{mse} of distributed \ac{sqrss} to distributed \ac{cs} \cite{baron2009distributed}, in which the quantized values of compressed projections are aggregated by and recovered the central decoder. We consider the cases where the decoder  recovers the set of signals  using the \ac{qiht} method \cite{jacques2013quantized} as well as \ac{fista} \cite{beck2009fast}.
While more advanced schemes combining distributed \ac{cs} and vector quantization were proposed in \cite{leinonen2018distributed}, their complexity grows rapidly when $\lenX > 2$, and thus we focus on conventional distributed \ac{cs} with scalar quantization.

Observing Fig. \ref{fig:distortion}, we note that the proposed distributed quantization scheme notably outperforms techniques based on distributed \ac{cs}. In particular, our method is shown to improve substantially the accuracy of the overall digital representation as the quantization rate increases, while distributed quantized \ac{cs} is demonstrated to meet an error floor around $4\cdot 10^{-2}$ for \ac{fista} and $9\cdot 10^{-3}$ for \ac{qiht}. \off{\textcolor{red}{Standard uniform quantization, which is applicable only for $\Rate > 1$ as the \acp{adc} must utilize at least one bit, is notably outperformed by  the previous approaches, as it does not exploit the underlying sparsity.}\footnote{\textcolor{red}{Alejandro - where did this statement come from? You  neither depict standard uniform quantization nor rates above 1... Is this a copy paste issue?}}}

\begin{figure}
	\centering
	{\includegraphics[trim=0cm 0.0cm 0cm 0cm, width = \columnwidth]{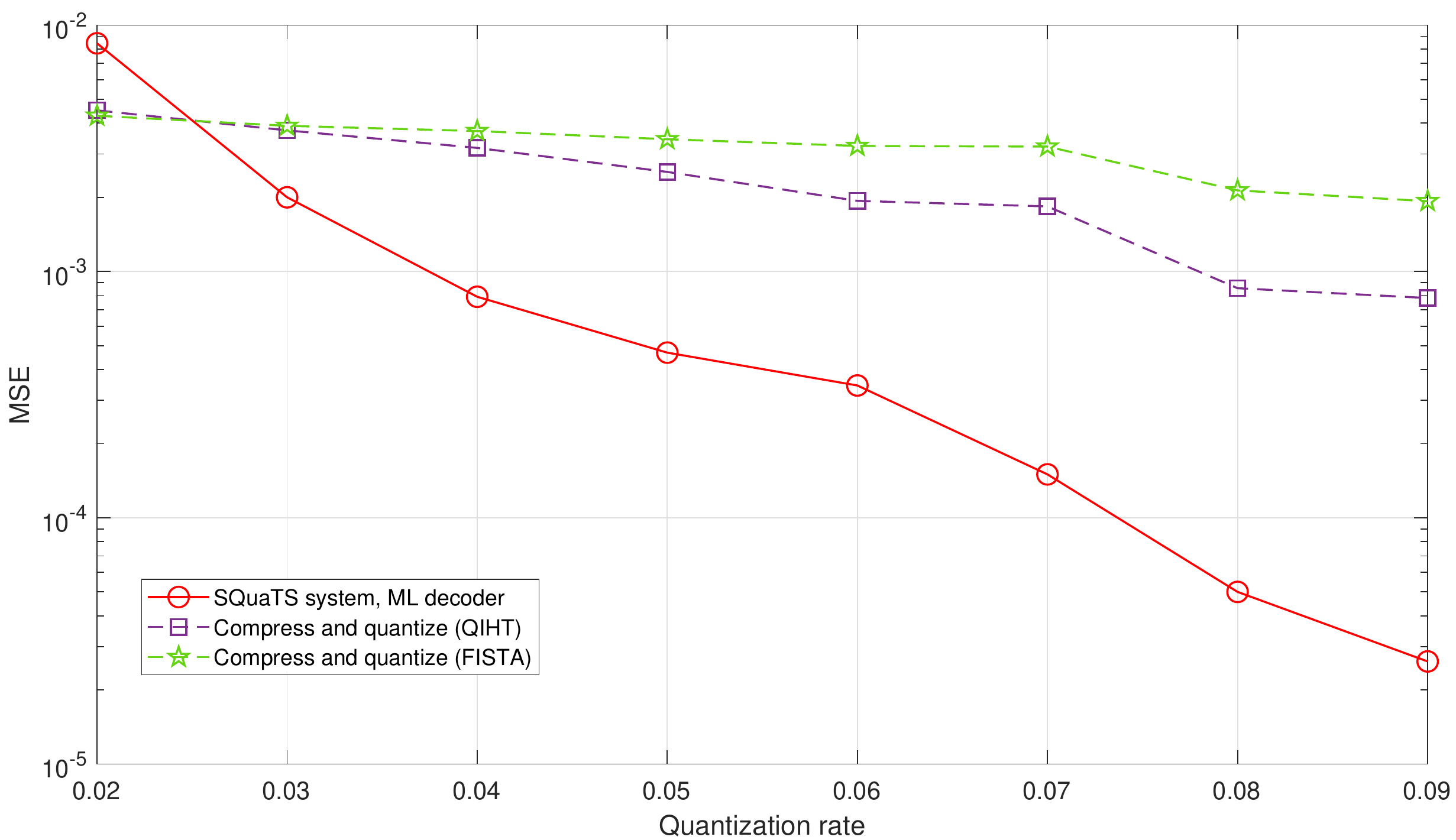}}
	\caption{Distributed quantization systems comparison, $\lenT = 100$, $n=10$, $\SpaSize = 3$.}
	\label{fig:distortion}
\end{figure}

Finally,  we demonstrate how the minimal quantization rate grows with the resolution $\lenL$. To that aim, we compute in Fig. \ref{fig:bound} the minimal rate versus $\lenL$. The setup evaluated here consists of $\lenX= 10$  sequences of $\lenT = 90$ samples each, for both overall sparsity with $\SpaSize \in \{6,12,24,36\}$ as well as structured sparsity with the same overall sparsity level and $k_s = 3$. Observing Fig. \ref{fig:bound}, we note that structured sparsity allows to use lower quantization rates, i.e., fewer bits, to achieve the same level of distortion, due to the additional structure. We also note that the quantization rate grows slowly with $\lenL$, indicating that a minor increase in the quantization rate can allow the scheme to utilize \acp{adc} of much higher resolution, while maintaining the guaranteed performance.
Observing Fig. \ref{fig:bound}, we note that structured sparsity allows to use lower quantization rates, i.e., fewer bits, to achieve the same level of distortion.

\begin{figure}
	\centering
	{\includegraphics[trim=0cm 0.0cm 0cm 0cm, width = \columnwidth]{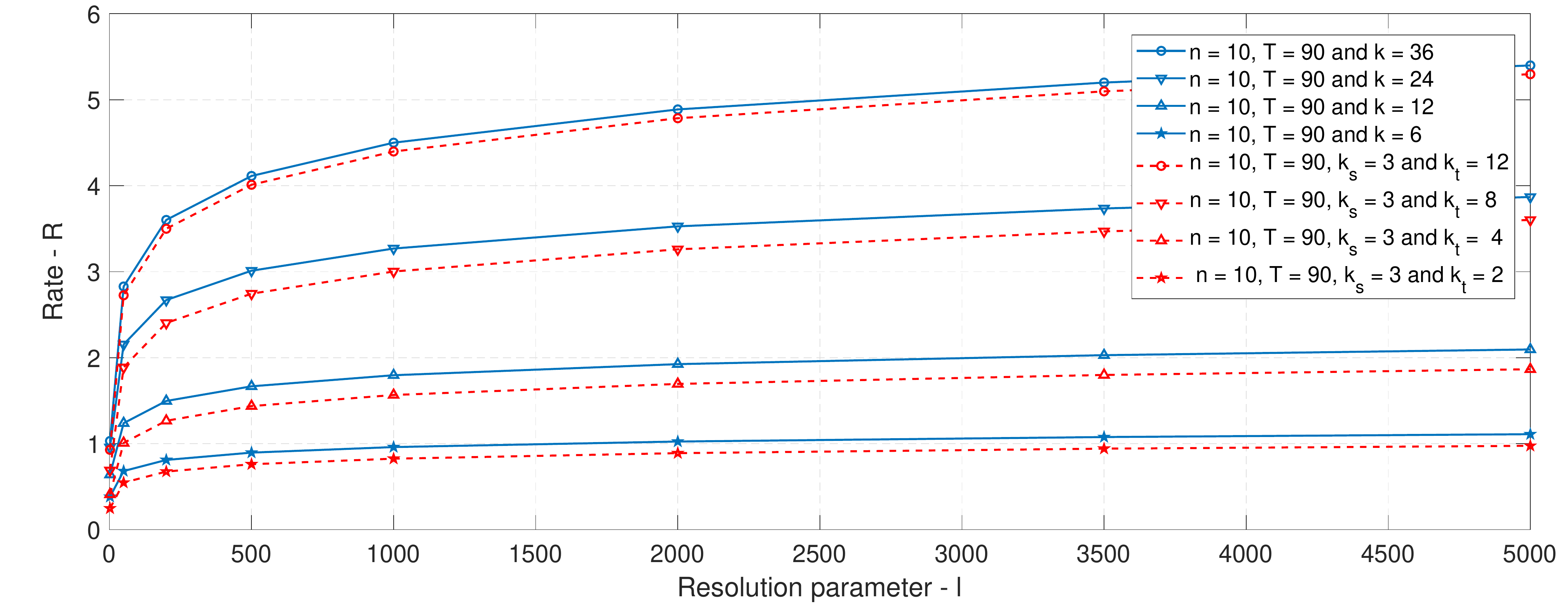}}
	\caption{Quantization rate threshold versus the resolution  $l$.}
	\label{fig:bound}
\end{figure}

	The results presented in this section demonstrate the potential of \ac{sqrss} as a quantization scheme for sparse signals which is both accurate as well as suitable for implementation with conventional serial scalar \acp{adc}, Our results also demonstrate the ability of \ac{sqrss} to implement distributed quantization and its robustness to digital noise.

	\section{Conclusions}\label{sec:conclusions}
	In this paper we proposed \ac{sqrss}, a quantization system designed for representing sparse signals acquired in a sequential manner. \ac{sqrss} combines code structures from group testing theory with the limitations and characteristics of conventional \acp{adc}. We derived the achievable \ac{mse} of the proposed scheme in the asymptotic signal size regime and characterized its  complexity.
	\ifFullVersion
	We proposed a reduced complexity decoding method for \ac{sqrss} which trades performance for computational burden, while maintaining the sequential acquisition property of \ac{sqrss}, and showed how \ac{sqrss} can be extended to distributed setups.
	\fi 
	Our simulation study demonstrates the substantial performance gain of \ac{sqrss} compared to directly applying a serial scalar \ac{adc}, as well as to \ac{cs}-based methods.
	
	
	\begin{appendix}
		%
		To prove the \Cref{direct theorem}, we first  provide a reliability bound which guarantees accurate reconstruction of the quantized representation $\{ Q(s[i])\}_{i=1}^{\lenT}$ from $\myY_{\lenT} $. Then, we show that this bound results in the condition on the quantization rate stated in \Cref{direct theorem}. An achievability bound on the required number of bits  is stated in the following lemma:
		\begin{lemma}\label{direct lemma1}  If for some $\varepsilon > 0$  independent of $\lenT$ and $\SpaSize$, the number of bits used for  digital representation  satisfies
			\begin{eqnarray}\label{eq:reduce_hw2}
			b \ge \max_{1 \leq i \leq k}\frac{(1+\varepsilon)k}{i}\log\binom{\lenT-k}{i}l^i,
			\end{eqnarray}
			then, under the code construction of Section \ref{sec:sqrss}, as $\lenT\rightarrow \infty$ the average error probability to recover $\{ Q(s[i])\}_{i=1}^{\lenT}$, given by $\frac{1}{\lenT} \sum\limits_{i=1}^{\lenT}\Pr\left(\hat{s}_i \neq Q(s[i]) \right)$, approaches zero exponentially.
		\end{lemma}
		\begin{proof}
			The Lemma follows from \cite[Lemma 1]{wsn2017drivejornal}, whose proof  is based on \cite[Lemma 2]{9218939}.
		\end{proof}
		We note that the corresponding bound in \cite[Theorem III.1]{atia2012boolean}, which studied group testing over a binary field, can be considered as a special case of Lemma \ref{direct lemma1} with $\lenL =1$, i.e., using one bit quantizers.
		In particular, since we consider quantizers with arbitrary resolution, the bound in \Cref{direct lemma1} must account for  the fact that  the codewords have to be selected from different bins, as $\lenL$ can be larger than one. 
		
		Now, \Cref{direct lemma1} yields a sufficient condition for the digital representation $\{\hat{s}[i]\}$ to approach the directly quantized $\{{s}[i]\}$, for which the \ac{mse} is $	D_{\lenT}(\lenL) $ given in \eqref{eqn:ScaMSE2}. Since  the quantization rate is given by $\Rate = \frac{\bits}{\lenT}$,  the condition \eqref{eq:reduce_hw2} becomes
		\begin{equation*}
		\Rate \ge \max_{1 \leq i \leq \SpaSize}\frac{(1+\varepsilon)k}{i\cdot \lenT}\log\left( \binom{\lenT-\SpaSize}{i}\cdot \lenL^i\right),
		\end{equation*} 	
		proving the theorem.
		\qed
	\end{appendix}
	
	\bibliographystyle{IEEEtran}
	\bibliography{references,SecureNetworkCodingGossip}

\begin{thebibliography}{10}
\providecommand{\url}[1]{#1}
\csname url@samestyle\endcsname
\providecommand{\newblock}{\relax}
\providecommand{\bibinfo}[2]{#2}
\providecommand{\BIBentrySTDinterwordspacing}{\spaceskip=0pt\relax}
\providecommand{\BIBentryALTinterwordstretchfactor}{4}
\providecommand{\BIBentryALTinterwordspacing}{\spaceskip=\fontdimen2\font plus
\BIBentryALTinterwordstretchfactor\fontdimen3\font minus
  \fontdimen4\font\relax}
\providecommand{\BIBforeignlanguage}[2]{{%
\expandafter\ifx\csname l@#1\endcsname\relax
\typeout{** WARNING: IEEEtran.bst: No hyphenation pattern has been}%
\typeout{** loaded for the language `#1'. Using the pattern for}%
\typeout{** the default language instead.}%
\else
\language=\csname l@#1\endcsname
\fi
#2}}
\providecommand{\BIBdecl}{\relax}
\BIBdecl

\bibitem{gray1998quantization}
R.~M. Gray and D.~L. Neuhoff, ``Quantization,'' vol.~44, no.~6, pp. 2325--2383,
  1998.

\bibitem{C10}
T.~M. Cover and J.~A. Thomas, \emph{Elements of information theory}.\hskip 1em
  plus 0.5em minus 0.4em\relax John Wiley \& Sons, 2012.

\bibitem{kosonocky1999analog}
S.~Kosonocky and P.~Xiao, ``Analog-to-digital conversion architectures,''
  \emph{Digital Signal Processing Handbook}, 1999.

\bibitem{polyanskiy2014lecture}
Y.~Polyanskiy and Y.~Wu, ``Lecture notes on information theory,'' \emph{Lecture
  Notes for ECE563 (UIUC) and}, vol.~6, pp. 2012--2016, 2014.

\bibitem{bhatt2018information}
A.~Bhatt, B.~Nazer, O.~Ordentlich, and Y.~Polyanskiy, ``Information-distilling
  quantizers,'' \emph{arXiv preprint arXiv:1812.03031}, 2018.

\bibitem{barnes2019learning}
L.~P. Barnes, Y.~Han, and A.~Ozgur, ``Learning distributions from their samples
  under communication constraints,'' \emph{arXiv preprint arXiv:1902.02890},
  2019.

\bibitem{shlezinger2018hardware}
N.~Shlezinger, Y.~C. Eldar, and M.~Rodrigues, ``Hardware-limited task-based
  quantization,'' vol.~67, no.~20, pp. 5223--5238, 2019.

\bibitem{shlezinger2018asymptotic}
N.~Shlezinger, Y.~C. Eldar, and M.~R. Rodrigues, ``Asymptotic task-based
  quantization with application to massive {MIMO},'' vol.~67, no.~15, pp.
  3995--4012, 2019.

\bibitem{salamatian2019task}
S.~Salamatian, N.~Shlezinger, Y.~C. Eldar, and M.~M{\'e}dard, ``Task-based
  quantization for recovering quadratic functions using principal inertia
  components,'' in \emph{Proc. IEEE ISIT}, 2019.

\bibitem{shlezinger2019deep}
N.~Shlezinger and Y.~C. Eldar, ``Deep task-based quantization,'' \emph{arXiv
  preprint arXiv:1908.06845}, 2019.

\bibitem{gubner1993distributed}
J.~A. Gubner, ``Distributed estimation and quantization,'' vol.~39, no.~4, pp.
  1456--1459, 1993.

\bibitem{lam1993design}
W.-M. Lam and A.~R. Reibman, ``Design of quantizers for decentralized
  estimation systems,'' vol.~41, no.~11, pp. 1602--1605, 1993.

\bibitem{berger1996ceo}
T.~Berger, Z.~Zhang, and H.~Viswanathan, ``The {CEO} problem [multiterminal
  source coding],'' vol.~42, no.~3, pp. 887--902, 1996.

\bibitem{oohama1998rate}
Y.~Oohama, ``The rate-distortion function for the quadratic gaussian ceo
  problem,'' vol.~44, no.~3, pp. 1057--1070, 1998.

\bibitem{el2011network}
A.~El~Gamal and Y.-H. Kim, \emph{Network information theory}.\hskip 1em plus
  0.5em minus 0.4em\relax Cambridge university press, 2011.

\bibitem{shlezinger2019joint}
N.~Shlezinger, S.~Salamatian, Y.~C. Eldar, and M.~M{\'e}dard, ``Joint sampling
  and recovery of correlated sources,'' in \emph{Proc. IEEE ISIT}, 2019.

\bibitem{saxena2006efficient}
A.~Saxena, J.~Nayak, and K.~Rose, ``On efficient quantizer design for robust
  distributed source coding,'' in \emph{Proc. IEEE DCC}, 2006, pp. 63--72.

\bibitem{wernersson2009distributed}
N.~Wernersson, J.~Karlsson, and M.~Skoglund, ``Distributed quantization over
  noisy channels,'' vol.~57, no.~6, pp. 1693--1700, 2009.

\bibitem{fleming2004network}
M.~Fleming, Q.~Zhao, and M.~Effros, ``Network vector quantization,'' vol.~50,
  no.~8, pp. 1584--1604, 2004.

\bibitem{wagner2012compressed}
N.~Wagner, Y.~C. Eldar, and Z.~Friedman, ``Compressed beamforming in ultrasound
  imaging,'' vol.~60, no.~9, pp. 4643--4657, 2012.

\bibitem{shechtman2014gespar}
Y.~Shechtman, A.~Beck, and Y.~C. Eldar, ``{GESPAR}: Efficient phase retrieval
  of sparse signals,'' vol.~62, no.~4, pp. 928--938, 2014.

\bibitem{rossi2014spatial}
M.~Rossi, A.~M. Haimovich, and Y.~C. Eldar, ``Spatial compressive sensing for
  {MIMO} radar,'' vol.~62, no.~2, pp. 419--430, 2014.

\bibitem{berger2010application}
C.~R. Berger, Z.~Wang, J.~Huang, and S.~Zhou, ``Application of compressive
  sensing to sparse channel estimation,'' vol.~48, no.~11, pp. 164--174, 2010.

\bibitem{feizi2011power}
S.~Feizi and M.~M{\'e}dard, ``A power efficient sensing/communication scheme:
  Joint source-channel-network coding by using compressive sensing,'' in
  \emph{Allerton Conference on Communication, Control, and Computing}, 2011,
  pp. 1048--1054.

\bibitem{eldar2012compressed}
Y.~C. Eldar and G.~Kutyniok, \emph{Compressed sensing: theory and
  applications}.\hskip 1em plus 0.5em minus 0.4em\relax Cambridge University
  Press, 2012.

\bibitem{duarte2011structured}
M.~F. Duarte and Y.~C. Eldar, ``Structured compressed sensing: From theory to
  applications,'' vol.~59, no.~9, pp. 4053--4085, 2011.

\bibitem{jacques2013robust}
L.~Jacques, J.~N. Laska, P.~T. Boufounos, and R.~G. Baraniuk, ``Robust 1-bit
  compressive sensing via binary stable embeddings of sparse vectors,''
  vol.~59, no.~4, pp. 2082--2102, 2013.

\bibitem{boufounos20081}
P.~T. Boufounos and R.~G. Baraniuk, ``1-bit compressive sensing,'' in
  \emph{Proc. IEEE CISS}, 2008, pp. 16--21.

\bibitem{jacques2011dequantizing}
L.~Jacques, D.~K. Hammond, and J.~M. Fadili, ``Dequantizing compressed sensing:
  When oversampling and non-gaussian constraints combine,'' vol.~57, no.~1, pp.
  559--571, 2011.

\bibitem{gunturk2010sigma}
C.~S. G{\"u}nt{\"u}rk, M.~Lammers, A.~Powell, R.~Saab, and {\"O}.~Yilmaz,
  ``Sigma delta quantization for compressed sensing,'' in \emph{Proc. IEEE
  CISS}, 2010.

\bibitem{kipnis2018single}
A.~Kipnis, G.~Reeves, and Y.~C. Eldar, ``Single letter formulas for quantized
  compressed sensing with gaussian codebooks,'' in \emph{Proc. IEEE ISIT},
  2018, pp. 71--75.

\bibitem{boufounos2015quantization}
P.~T. Boufounos, L.~Jacques, F.~Krahmer, and R.~Saab, ``Quantization and
  compressive sensing,'' in \emph{Compressed sensing and its
  applications}.\hskip 1em plus 0.5em minus 0.4em\relax Springer, 2015, pp.
  193--237.

\bibitem{saab2018quantization}
R.~Saab, R.~Wang, and {\"O}.~Y{\i}lmaz, ``Quantization of compressive samples
  with stable and robust recovery,'' \emph{Applied and Computational Harmonic
  Analysis}, vol.~44, no.~1, pp. 123--143, 2018.

\bibitem{sarvotham2005distributed}
S.~Sarvotham, D.~Baron, M.~Wakin, M.~F. Duarte, and R.~G. Baraniuk,
  ``Distributed compressed sensing of jointly sparse signals,'' in
  \emph{Asilomar conference on signals, systems, and computers}, 2005, pp.
  1537--1541.

\bibitem{baron2009distributed}
D.~Baron, M.~F. Duarte, M.~B. Wakin, S.~Sarvotham, and R.~G. Baraniuk,
  ``Distributed compressive sensing,'' \emph{arXiv preprint arXiv:0901.3403},
  2009.

\bibitem{do2009distributed}
T.~T. Do, Y.~Chen, D.~T. Nguyen, N.~Nguyen, L.~Gan, and T.~D. Tran,
  ``Distributed compressed video sensing,'' in \emph{Proc. IEEE ICIP}, 2009,
  pp. 1393--1396.

\bibitem{patterson2014distributed}
S.~Patterson, Y.~C. Eldar, and I.~Keidar, ``Distributed compressed sensing for
  static and time-varying networks,'' vol.~62, no.~19, pp. 4931--4946, 2014.

\bibitem{feizi2010compressive}
S.~Feizi, M.~M{\'e}dard, and M.~Effros, ``Compressive sensing over networks,''
  in \emph{Allerton Conference on Communication, Control, and Computing}, pp.
  1129--1136.

\bibitem{shirazinia2014distributed}
A.~Shirazinia, S.~Chatterjee, and M.~Skoglund, ``Distributed quantization for
  measurement of correlated sparse sources over noisy channels,'' \emph{arXiv
  preprint arXiv:1404.7640}, 2014.

\bibitem{leinonen2018distributed}
M.~Leinonen, M.~Codreanu, and M.~Juntti, ``Distributed distortion-rate
  optimized compressed sensing in wireless sensor networks,'' vol.~66, no.~4,
  pp. 1609--1623, 2018.

\bibitem{9218939}
A.~{Cohen}, A.~{Cohen}, and O.~{Gurewitz}, ``Secure group testing,'' \emph{IEEE
  Transactions on Information Forensics and Security}, pp. 1--1, 2020.

\bibitem{li1999asymptotic}
J.~Li, N.~Chaddha, and R.~M. Gray, ``Asymptotic performance of vector
  quantizers with a perceptual distortion measure,'' vol.~45, no.~4, pp.
  1082--1091, 1999.

\bibitem{eldar2015sampling}
Y.~C. Eldar, \emph{Sampling theory: Beyond bandlimited systems}.\hskip 1em plus
  0.5em minus 0.4em\relax Cambridge University Press, 2015.

\bibitem{wsn2017drivejornal}
A.~{Cohen}, A.~{Cohen}, and O.~{Gurewitz}, ``{Efficient Data Collection Over
  Multiple Access Wireless Sensors Network},'' \emph{IEEE/ACM Transactions on
  Networking}, vol.~28, no.~2, pp. 491--504, 2020.

\bibitem{dorfman1943detection}
R.~Dorfman, ``The detection of defective members of large populations,''
  \emph{The Annals of Mathematical Statistics}, vol.~14, no.~4, pp. 436--440,
  1943.

\bibitem{panter1951quantization}
P.~Panter and W.~Dite, ``Quantization distortion in pulse-count modulation with
  nonuniform spacing of levels,'' \emph{Proceedings of the IRE}, vol.~39,
  no.~1, pp. 44--48, 1951.

\bibitem{chan2014non}
C.~L. Chan, S.~Jaggi, V.~Saligrama, and S.~Agnihotri, ``Non-adaptive group
  testing: Explicit bounds and novel algorithms,'' vol.~60, no.~5, pp.
  3019--3035, 2014.

\bibitem{ziv1985universal}
J.~Ziv, ``On universal quantization,'' vol.~31, no.~3, pp. 344--347, 1985.

\bibitem{zamir1992universal}
R.~Zamir and M.~Feder, ``On universal quantization by randomized
  uniform/lattice quantizers,'' vol.~38, no.~2, pp. 428--436, 1992.

\bibitem{macula1999probabilistic}
A.~J. Macula, ``Probabilistic nonadaptive group testing in the presence of
  errors and {DNA} library screening,'' \emph{Annals of Combinatorics}, vol.~3,
  no.~1, pp. 61--69, 1999.

\bibitem{damaschke2010bounds}
P.~Damaschke and A.~Muhammad, ``Bounds for nonadaptive group tests to estimate
  the amount of defectives,'' \emph{Combinatorial Optimization and
  Applications}, pp. 117--130, 2010.

\bibitem{damaschke2010competitive}
P.~Damaschke and A.~S. Muhammad, ``Competitive group testing and learning
  hidden vertex covers with minimum adaptivity,'' \emph{Discrete Mathematics,
  Algorithms and Applications}, vol.~2, no.~03, pp. 291--311, 2010.

\bibitem{atia2012boolean}
G.~K. Atia and V.~Saligrama, ``Boolean compressed sensing and noisy group
  testing,'' vol.~58, no.~3, pp. 1880--1901, 2012. A minor corection appered in
  vol. 61, no. 3, pp. 1507-1507, 2015.

\bibitem{emad2014poisson}
A.~Emad and O.~Milenkovic, ``Poisson group testing: A probabilistic model for
  nonadaptive streaming boolean compressed sensing,'' in \emph{Proc. IEEE
  ICASSP}, 2014, pp. 3335--3339.

\bibitem{aldridge2014group}
M.~Aldridge, L.~Baldassini, and O.~Johnson, ``Group testing algorithms: bounds
  and simulations,'' vol.~60, no.~6, pp. 3671--3687, 2014.

\bibitem{coja2019information}
A.~Coja-Oghlan, O.~Gebhard, M.~Hahn-Klimroth, and P.~Loick,
  ``Information-theoretic and algorithmic thresholds for group testing,''
  \emph{arXiv preprint arXiv:1902.02202}, 2019.

\bibitem{bui2019efficient}
T.~V. Bui, M.~Kuribayashi, T.~Kojima, R.~Haghvirdinezhad, and I.~Echizen,
  ``Efficient (nonrandom) construction and decoding for non-adaptive group
  testing,'' \emph{Journal of Information Processing}, vol.~27, pp. 245--256,
  2019.

\bibitem{cohen2020multi}
A.~Cohen, N.~Shlezinger, A.~Solomon, Y.~C. Eldar, and M.~M{\'e}dard,
  ``Multi-level group testing with application to one-shot pooled covid-19
  tests,'' \emph{arXiv preprint arXiv:2010.06072}, 2020.

\bibitem{dantzig2003max}
G.~Dantzig and D.~R. Fulkerson, ``On the max flow min cut theorem of
  networks,'' \emph{Linear inequalities and related systems}, vol.~38, pp.
  225--231, 2003.

\bibitem{jacques2013quantized}
L.~Jacques, K.~Degraux, and C.~De~Vleeschouwer, ``Quantized iterative hard
  thresholding: Bridging 1-bit and high-resolution quantized compressed
  sensing,'' \emph{arXiv preprint arXiv:1305.1786}, 2013.

\bibitem{beck2009fast}
A.~Beck and M.~Teboulle, ``A fast iterative shrinkage-thresholding algorithm
  for linear inverse problems,'' \emph{SIAM journal on imaging sciences},
  vol.~2, no.~1, pp. 183--202, 2009.

\bibitem{sejdinovic2010note}
D.~Sejdinovic and O.~Johnson, ``Note on noisy group testing: asymptotic bounds
  and belief propagation reconstruction,'' in \emph{Allerton Conference on
  Communication, Control, and Computing}, 2010, pp. 998--1003.

\bibitem{chan2011non}
C.~L. Chan, P.~H. Che, S.~Jaggi, and V.~Saligrama, ``Non-adaptive probabilistic
  group testing with noisy measurements: Near-optimal bounds with efficient
  algorithms,'' in \emph{Allerton Conference on Communication, Control, and
  Computing}, 2011, pp. 1832--1839.

\end{thebibliography}
\end{document}